\documentclass[11pt]{article}
\usepackage[left=.6in,top=.6in,right=.6in,bottom=.6in,head=0in,nofoot]{geometry}
\usepackage{tikz}
\usetikzlibrary{matrix}

\usepackage{setspace,url,bm,amsmath} 
\usepackage{scalerel}

\usepackage{xcolor}
\usepackage{multirow}
\usepackage{arydshln}
\usepackage{bbm}

\usepackage{titlesec} 
\titlelabel{\thetitle.\quad} 
\titleformat*{\section}{\bf\Large\center}

\usepackage{eufrak}
\usepackage{float}
\usepackage{graphicx} 
\usepackage{bbm}
\usepackage{latexsym}
\usepackage{caption}
\usepackage[margin=20pt]{subcaption}
\usepackage{enumerate}
\usepackage{mathtools}

\graphicspath{ {./Graph/} }
\usepackage{xr}
 
\usepackage{amsthm}
\usepackage{amssymb}
\usepackage{amsmath}
\usepackage{color}

\usepackage{comment}
\theoremstyle{definition}
\newtheorem{assumption}{Assumption}
\newtheorem*{theorem*}{Theorem}
\newtheorem{theorem}{Theorem}
\newtheorem*{rmk*}{remark}
\newtheorem{proposition}{Proposition}
\newtheorem{lemma}{Lemma}

\newtheorem{algo}{Algorithm}

\newtheorem*{corollary*}{Corollary}

\usepackage{natbib} 
\bibpunct{(}{)}{;}{a}{}{,} 

\usepackage{etoolbox} 
\apptocmd{\sloppy}{\hbadness 10000\relax}{}{} 

\usepackage{multibib}
\newcites{sec}{References}

\usepackage{color}
\usepackage{listings}
\usepackage{hyperref}
\usepackage{booktabs}

\usepackage{lscape}

\usepackage[textsize=tiny, textwidth = 2cm, shadow]{todonotes}

\def\xeps{\textup{ps}}
\def\xps{\textup{ps}}
\def\xreg{\textup{reg}}
\def\hyua{\hy_\textup{unadj}}
\def\proj{\textup{proj}}
\def\htn{\htau_\ua}
\def\hteps{\htau_\ipwx}
\def\htreg{\htau_\wlsx}

\def\ols{{ordinary-least-squares }}
\def\wls{{weighted-least-squares }}

\def\yizz{Y_i(z)}

\def\dmd{\diamond}
\def\iptwn{{propensity score weighting}}
\def\iptw{{propensity score weighting} }
\def\hajx{\textup{x-ps}}

\def\sub{\textup{sub}}
\def\aug{\textup{mim}}
\def\ua{\textup{unadj}}
\def\hajek{{H\'{a}jek} }

\def\wiz{w_i^0}
\def\wic{w_i^{(c)}}

\def\muo{\mu_1}
\def\muz{\mu_0}
\def\muzz{\mu_z}
\def\go{\gamma_1}
\def\gz{\gamma_0}
\def\hmo{\hat m_1}
\def\hmz{\hat m_0}
\def\sumi{\sum_{i=1}^N}
\def\meani{N^{-1}\sumi}

\def\muxst{\mu_x^{*\T}}

\def\riw{R_i^w}

\def\mim{\textup{mim}}
\def\imp{\textup{imp}}

\def\hydro{\hy_\dr(1)}
\def\hydrz{\hy_\dr(0)}

\def\hywlsx{\hy_\wlsx}
\def\gos{\gamma_1^*}
\def\gzs{\gamma_0^*}
\def\gost{\gamma_1^{*\T}}
\def\gzst{\gamma_0^{*\T}}

\def\bsm{\quad\text{by symmetry}}
\def\blb{\quad\text{by Lemma \ref{lem:beta}}}
 
\def\covxinv{ \{\cov(x_i) \}^{-1}}

\def\abax{A^{*-1}B^*(A^{*-1})^\T} 
\def\ghg{G^{*-1}H^*(G^{*-1})^\T} 
\def\aba{A_\ua^{*-1}B_\ua^*(A_\ua^{*-1})^\T} 
\def\abaeps{A_\xeps^{*-1}B_\xeps^*(A_\xeps^{*-1})^\T}

\def\abareg{A_\xreg^{*-1}B_\xreg^*(A_\xreg^{*-1})^\T}

\def\muos{\mu_1^*}
\def\muzs{\mu_0^*}
\def\muxs{\mu_x^*}

\def\hdo{\Delta\hat x(1)}
\def\hdz{\Delta\hat x(0)}
\def\yio{Y_i(1)}
\def\yiz{Y_i(0)}
\def\hgo{\hg_1}
\def\hgz{\hg_0}

\def\mux{\mu_x}

\def\dyoa{ \left\{Y_i(1)- \dxi ^\T  \gamma^*_1 - \mu^*_1 \right\}}
\def\dyza{\left\{Y_i(0)- \dxi ^\T  \gamma^*_0 - \mu^*_0 \right\}}
\def\tdxit{\left(1, \dxit \right)}
\def\tdxi{\beginp 
1\\ x_i - \muxs
\endp}

\def\dyo{\{Y_i(1) - \mu_1^*\}} 
\def\dyz{\{Y_i(0) - \mu_0^*\}} 
\def\dyon{\{Y_i(1) - \mu_1\}} 
\def\dyzn{\{Y_i(0) - \mu_0\}}

\def\dxi{(x_i - \muxs)}
\def\dxit{\dxi^\T}

\def\ryi{\riy}

\def\va{\var_\textup{a}}
\def\dr{\textup{dr}}
\def\cova{\cov_\textup{a}}
\def\rpn{\dfrac{\riy}{p_i}}
\def\rpfun{\dfrac{\riy}{\pfun}}

\def\hpi{\hat p_i}
\def\hy{\hat Y}
\def\hmuox{\hy_\wlsx(1)}
\def\hmuzx{\hy_\wlsx(0)}

\def\hywlsxo{\hy_\wlsx(1)}
\def\hywlsxz{\hy_\wlsx(0)}

\def\hywlsxop{\hy_\wlsx'(1)}
\def\hywlsxzp{\hy_\wlsx'(0)}

\def\hyipwo{\hy_\ipw(1)}
\def\hyipwz{\hy_\ipw(0)}

\def\hyipwxo{\hy_\ipwx(1)}
\def\hyipwxz{\hy_\ipwx(0)}

\def\hmuzwlsx{\hy_\wlsx(0)}

\def\ef{e(x_i; \alpha)}
\def\hei{\hat e_i}
\def\zhe{\dfrac{Z_i}{\hei}}
\def\omzhe{\dfrac{1-Z_i}{1-\hei}}
\def\ze{\dfrac{Z_i}{e}}
\def\zef{\dfrac{Z_i}{e(x_i; \alpha)}}
\def\omzef{\dfrac{1-Z_i}{1-e(x_i; \alpha)}}
\def\omze{\dfrac{1-Z_i}{1-e}}

\def\hmuo{\hyua(1)}
\def\hmuz{\hyua(0)}

\def\hmuoipw{\hyua(1)}
\def\hmuzipw{\hyua(0)}

\def\hmuoipwx{\hat Y_\ipwx(1)}
\def\hmuzipwx{\hat Y_\ipwx(0)}

\def\rp{\dfrac{\riy}{\hpi}}
\def\ipwx{\textup{ps}}
\def\wlsx{\textup{reg}}
\def\xreg{\textup{reg}}
\def\hpii{\hat p_i^{-1}}
\def\riy{R_i^Y}
\def\pfun{p(x_i, Z_i; \beta)}

\def\ipw{\textup{unadj}}
\def\mim{\textup{mim}}
\def\imp{\textup{imp}}

\def\htau{\hat\tau}

\def\tyi{\tilde Y_i}

\def\beginp{\begin{pmatrix}}
\def\endp{\end{pmatrix}}

\newcommand{\mn}{\mathcal N}

\newcommand{\op}{o_{\pr}(1)}
\newcommand{\rs}{\rightsquigarrow}

\newcommand{\glmt}{\texttt{glm}}
\newcommand{\lmt}{\texttt{lm}}

\def\hg{\hat\gamma}

\def\mn{\mathcal N}

\def\beginy{\begin{eqnarray}}
\def\endy{\end{eqnarray}}

\def\T{{ \mathrm{\scriptscriptstyle T} }}
\def\tx{\tilde x}
\def\txi{\tx_i}

\def\begina{\begin{eqnarray*}}
\def\enda{\end{eqnarray*}}

\newcommand{\sm}{Supplementary Material}
\newcommand{\sms}{{Supplementary Material }}

\def\begini{\begin{itemize}}
\def\endi{\end{itemize}}
\def\begine{\begin{enumerate}}
\def\ende{\end{enumerate}}

\newcommand{\pd}[1]{\dfrac{\partial}{\partial #1}}

\newcommand{\ot}[1]{1, \ldots, #1}

\newcommand{\indep}{\perp \!\!\! \perp}

\RequirePackage[normalem]{ulem}

\allowdisplaybreaks

\newcommand{\pr}{\mathbb P}
\DeclareMathOperator\cov{cov}    
\DeclareMathOperator\var{var}

 \newcommand{\footremember}[2]{%
    \footnote{#2}
    \newcounter{#1}
    \setcounter{#1}{\value{footnote}}%
}

\begin{document}

\onehalfspacing

\title{Covariate adjustment in randomized experiments with missing outcomes and covariates}

\author{Anqi Zhao\footremember{a}{Fuqua School of Business, Duke University, az171@duke.edu}
\and 
Peng Ding\footremember{b}{Department of Statistics, University of California, Berkeley,  pengdingpku@berkeley.edu}
\and Fan Li\footremember{c}{Department of Statistical Science, Duke University,   
 fl35@duke.edu \\\\
Ding thanks the U.S. National Science Foundation grant \#1945136. Li thanks Patient-Centered Outcomes Research Institute contract ME-2019C1-16146. We thank Jerry Chang, Wang Rui, Sean O'Brien, Luke Miratrix, and participants of Harvard Data Science Initiative causal seminars for stimulating discussions, the associate editor and two reviewers for constructive comments, and the investigators of the Best Apnea Interventions for Research trial for providing the data.}}
\date{}

\maketitle

\begin{abstract}
Covariate adjustment can improve precision in analyzing randomized experiments. With fully observed data, regression adjustment and \iptw are asymptotically equivalent in improving efficiency over unadjusted analysis.
When some outcomes are missing, we consider combining these two adjustment methods with inverse probability of observation weighting for handling missing outcomes, and show that the equivalence between the two methods breaks down. 
Regression adjustment no longer ensures efficiency gain over unadjusted analysis unless the true outcome model is linear in covariates or the outcomes are missing completely at random.
Propensity score weighting, in contrast, still guarantees efficiency over unadjusted analysis, and including more covariates in adjustment never harms asymptotic efficiency. 
Moreover, we establish the value of using partially observed covariates to secure additional efficiency by the missingness indicator method, which imputes all missing covariates by zero and uses the union of the completed covariates and corresponding missingness indicators as the new, fully observed covariates. 
Based on these findings, 
we recommend using regression adjustment in combination with  the missingness indicator method if the linear outcome model or missing complete at random assumption is plausible and using propensity score weighting with the missingness indicator method otherwise. 
\end{abstract}


\section{Covariate adjustment in randomized experiments: a review and open questions}\label{sec:review}
Adjusting for chance imbalance in covariates can improve precision in analyzing randomized experiments \citep{fisher1935design, lin2013agnostic}.
Consider a randomized controlled trial with two treatment levels of interest, indexed by $z =$ 1 for treatment and 0 for control,  and a study population of $N$ units, indexed by $i = \ot{N}$. 
Let $x_i \in \mathbb R^J$, $Z_i \in \{1,0\}$, and $Y_i \in \mathbb R$ denote the baseline covariates, treatment assignment, and outcome of unit $i$. 
Let $Y_i(1)\in\mathbb R$ and $Y_i(0)\in\mathbb R$ denote the potential outcomes of unit $i$ under treatment and control, respectively, with $Y_i = Z_i Y_i(1) + (1-Z_i) Y_i(0)$. 
Assume throughout that (i) the $N$ units are an independent and identically distributed sample from some population and (ii) the treatment levels are assigned independently across units with constant treatment probability $\pr(Z_i = 1) = e \in (0,1)$. 
Define $\tau_i = \yio - \yiz$ as the individual treatment effect for unit $i$ and $\tau = E(\tau_i) =  E\{Y_i(1)\}   - E\{Y_i(0)\}$
as the average treatment effect of interest. 
We first review below three estimators of $\tau$ when $(Y_i, x_i, Z_i)$ are fully observed.

A simple unbiased estimator of $\tau$ is the difference in means of the outcomes between the two treatment groups, commonly referred to as the unadjusted estimator and denoted by $\htau_{\ua}$. It is numerically equal to the coefficient of $Z_i$ from the \ols fit of the {\it unadjusted regression} of $Y_i$ on $(1, Z_i)$, denoted by $\lmt(Y_i \sim 1+Z_i)$ by R convention.

Regression adjustment and \iptw are two ways to adjust for chance imbalances in covariates. First, the {\it interacted regression} $
\lmt \{Y_i \sim 1 + Z_i + (x_i - \bar x) + Z_i(x_i - \bar x) \}$, 
where $\bar x = N^{-1}\sum_{i=1}^N x_i$, gives a covariate-adjusted variant of the unadjusted regression \citep{tsiatis2008covariate, lin2013agnostic, negi2020revisiting}.
The ordinary-least-squares coefficient of $Z_i$ defines a regression-adjusted estimator of $\tau$, denoted by $\htreg $. 

Next, let $e_i = \pr(Z_i = 1\mid x_i)$ denote the propensity score of unit $i$ \citep{rosenbaum1983central}. 
The \iptw approach to covariate adjustment weights observations by functions of an \emph{estimate} of $e_i$ \citep{williamson2014variance}. 
In our setting, the $e_i$ is known and equals $e$ for all units.
Nevertheless, standard results suggest that we can still estimate $e_i$ using a working model as a means to improve efficiency; see, e.g., \cite{hahn1998role, hirano2003efficient, shen2014inverse}.
Specifically, let $\hat e_i$ be the maximum likelihood estimate of $e_i$ based on the logistic regression of $Z_i$ on $(1, x_i)$, denoted by $\glmt(Z_i \sim 1 + x_i)$ by R convention. 
We can estimate $\tau$ by the coefficient of $Z_i$ from the \wls fit of the unadjusted regression $\lmt(Y_i \sim 1+Z_i)$, where we weight unit $i$ by $\hat e_i^{-1}$ if $Z_i = 1$ and by $(1-\hat e_i)^{-1}$ if $Z_i = 0$, summarized as $\hat \pi_i  =  Z_i / \hat e_i +(1-Z_i)/(1-\hei)$. 
We denote the resulting estimator by $\hteps$, where the ``-ps" stands for propensity score weighting. Other propensity score weights such as overlap weighting can also be used \citep{zeng2021propensity}.  

\begin{table}\caption{\label{tb:estimators} A summary of $\{ \htn , \htreg , \hteps  \}$ when all data are observed (column 3) and when outcomes are partially missing (column 4). 
Let $x_i' = x_i - \bar x$ denote the centered covariates, where $\bar x = \meani x_i$. Let $\hat p_i$ denote the estimated probability of $Y_i$ being observed given $(x_i, Z_i)$, and let $\hat \pi_i = Z_i / \hat e_i  + (1-Z_i) / (1-\hat e_i)$ denote the inverse of the estimated probability of the treatment received for unit $i$ with 
$\hat e_i$ as the estimated propensity score. 
}

 \resizebox{\columnwidth}{!}{%
 \begin{tabular}{l|c|c|c}\hline
 & Regression specification& Weight over $i = \ot{n}$ & Weight over $\{i:\riy = 1\}$ \\\hline
 $\htn $ & $\lmt(Y_i \sim 1+Z_i)$ & 1 &  $\hpi^{-1}$ \\\hline
 $\htreg $ & $\lmt(Y_i \sim 1+Z_i + x_i' + Z_i x_i')$ &  1 &   $\hpi^{-1}$ \\\hline
 $\hteps $ & $\lmt(Y_i \sim 1+Z_i)$ &  
 $\hat \pi_i$
 &  
 $\hpi^{-1} \hat \pi_i $\\\hline
 
 \end{tabular}
 }
 \end{table}

The $\{ \htn , \htreg , \hteps  \}$ together define three regression estimators of $\tau$ with fully observed data, summarized in the first three columns of Table \ref{tb:estimators}. 
Under mild regularity conditions, they are all consistent and asymptotically normal \citep{tsiatis2008covariate, lin2013agnostic, williamson2014variance, negi2020revisiting}, with $\htreg$ and $\hteps$ being asymptotically equivalent in improving precision over $\htn$ \citep{shen2014inverse,zeng2021propensity}; see Theorem {S1} in the \sms for a formal statement. 

Missing data are common in practice and pose challenges to inference. 
Assuming missingness only in covariates, \cite{zhao2022to} proposed to use the interacted regression with missingness indicators for covariates included as additional covariates, and showed that the resulting inference guarantees asymptotic efficiency over unadjusted analysis.   
Despite the vast literature on missing data and covariate adjustment separately, there lacks theoretical guidance on covariate adjustment with missingness in both outcomes and covariates in randomized experiments. Many important questions remain open;  for example, 
(i) How do we conduct covariate adjustment in the presence of missing outcomes? 
(ii) Does the resulting inference ensure consistency and efficiency gain over unadjusted analysis? (iii) Does the asymptotic equivalence between regression adjustment and \iptw still hold? 
(iv) Can the missingness indicator method in \cite{zhao2022to} be extended to the presence of missing outcomes? \cite{chang2023covariate} discussed some of these issues and proposed an estimator without theoretical investigation. This paper provides theoretical answers to these questions and 
proposes two easy-to-implement estimators.
We begin with the case with missingness in only outcomes in Section \ref{sec:missingY} and then extend to the case with missingness in both covariates and outcomes in Section \ref{sec:missingxy}.

\section{Covariate adjustment with missing outcomes }\label{sec:missingY}
\subsection{Regression estimators with missing outcomes}
We first extend Section \ref{sec:review} to the presence of missing outcomes. 
Assume throughout the rest of the paper that $x_i$ and $Z_i$ are fully observed for all units whereas $Y_i$ is missing for some units. 
Let $\ryi \in \{1,0\}$ be the indicator of $Y_i$ being observed for unit $i$, with $\ryi = 1$ if $Y_i$ is observed and $\ryi = 0$ otherwise. 
Recall from Table \ref{tb:estimators} that when all data are observed, $\htn$, $\htreg$, and $\hteps$ are the coefficients of $Z_i$ from the least-squares fits of the unadjusted regression, the interacted regression, and the unadjusted regression over all units, respectively, with weights $\pi_{i,\ua} = 1$, $\pi_{i,\wlsx} = 1$, and $\pi_{i,\ipwx} = \hat \pi_i $ for unit $i$.
In the presence of missing outcomes,
let $p_i = \pr(\ryi = 1\mid x_i, Z_i)$ denote the probability of $Y_i$ being observed given $(x_i, Z_i)$, and let $\hat p_i$ be an estimate of $p_i$. 
By inverse probability of observation weighting \citep{seaman2013review}, we can instead fit 
the corresponding regression over units with observed outcomes, indexed by $\{i: \riy =1\}$, with weight $\pi_{i, \dmd}' = \hpi^{-1}  \pi_{i,\dmd} \ (\dmd = \ua, \wlsx, \ipwx)$ for unit $i$.
This generalizes $\{\htn , \htreg , \hteps \}$ to the presence of missing outcomes, summarized in the last column of Table \ref{tb:estimators}. 
The definitions of $\htau_\dmd$'s when all data are observed are special cases with $\riy = 1$ and $\hpi = p_i = 1$ for all $i$. 
We will hence use the same notation to denote the generalized estimators with missing outcomes to highlight the connection. 
The generalized $\hteps$ is a double-weighted estimator, where we use $ \hpi^{-1}$ and $\hat\pi_i$ to address missing outcomes and covariate adjustment, respectively \citep{negi2020doubly,chang2023covariate}.

\subsection{Asymptotic theory}\label{sec:missingy}
We now establish the asymptotic properties of the generalized $\{\htn , \htreg , \hteps\}$. 
To begin with, Assumption \ref{assm:riy} specifies the outcome missingness mechanism.

\begin{assumption}\label{assm:riy}
\begine[(i)]
\item\label{item:mar} $\riy \indep \{\yio, \yiz\} \mid (x_i, Z_i)$; 
\item\label{item:p-model} $
p_i = \pr(R_i^Y = 1 \mid x_i, Z_i) =  \{1+\exp(-U_i^\T \beta^*)\}^{-1},
$
where $U_i = U(x_i, Z_i)$ is a known vector function of $(x_i, Z_i)$ and $\beta^*$ is the unknown parameter, and we construct $\hpi$ by the logistic regression $\glmt(\ryi \sim U_i)$ over $i = \ot{N}$.
\ende 
\end{assumption}

Assumption \ref{assm:riy}\eqref{item:mar} ensures that $\ryi$ is independent of $Y_i$ conditioning on the fully observed $(x_i, Z_i)$. 
The outcome is hence {\it missing at random} in the sense that whether an outcome is missing is independent of the value of the outcome conditional on the observables.
Assumption \ref{assm:riy}\eqref{item:p-model} further specifies the functional form of the outcome missingness mechanism.
We focus on the logistic missingness model because of its prevalence in practice.
We conjecture that similar results hold for general missingness models and relegate the formal theory to future research.
We use  $U_i$ to denote the regressor vector in the true outcome missingness model under Assumption \ref{assm:riy}\eqref{item:p-model}. 
In practice, the true value of $U_i$ is often unknown. 
Common choices for fitting a working model include $U_i = (1, x_i^\T)^\T$,  $U_i = (1,  x_i^\T, Z_i)^\T$, and $U_i = (1, x_i^\T, Z_i, Z_i x_i^\T)^\T$.

Theorem \ref{thm:clt} below generalizes the theory of covariate adjustment with fully observed data to the presence of missing outcomes and gives the asymptotic distributions of the generalized $\{\htn , \htreg , \hteps  \}$.
Let $\tyi = e^{-1} Y_i(1) +(1-e)^{-1} \yiz$, and let 
$
\textup{proj}(\tyi \mid 1, x_i) = E(\tyi) + \cov(\tyi, x_i) \covxinv \{x_i - E(x_i)\}
$
denote the linear projection of  $\tyi$ on $(1,x_i)$.

\begin{theorem}\label{thm:clt}
Assume complete randomization that $Z_i\indep \{ Y_i(1), Y_i(0), x_i\}$ and Assumption \ref{assm:riy}.
Under standard regularity conditions as $N\to\infty$, we have 
$
\sqrt N(\hat\tau_\dmd - \tau) \rightarrow  \mn\left(0,  v_\dmd \right)$ in distribution for $\dmd = \ua,  \wlsx, \hajx$, 
where \begine[(i)]
\item $
v_\ipwx = v_\ua - e(1-e)\var \{ \textup{proj}(\tyi \mid 1, x_i)  \} \leq v_\ua$; 
\item $v_\wlsx$ can be either greater or less than $v_\ua$ depending on the data generating process. As two special cases, we have $v_\wlsx \leq v_\ipwx \leq v_\ua$ if (a) $Y_i$ is missing completely at random with $p_i = p \in (0,1)$ or (b) the outcome model $
E(Y_i \mid x_i, Z_i=z)= E\{Y_i(z) \mid x_i\}$ is linear in $x_i$ for $z = 0,1$.
\ende
\end{theorem}

We relegate the explicit expressions of $v_\ua$ and $v_\wlsx$ to Theorem {S1} in the \sm. When $p_i = 1$ for all $i$, the three asymptotic variances reduce to those in the standard theory for fully observed data with $v_\wlsx = v_\ipwx \leq v_\ua$. 
Theorem \ref{thm:clt} has two implications.
First, it ensures the consistency and asymptotic normality of $\{\htn , \htreg , \hteps  \}$ in the presence of missing outcomes. 
Second, it clarifies the relative efficiency of $\{\htn , \htreg , \hteps  \}$ and highlights a key deviation from the theory when all outcomes are observed:
regression adjustment by the interacted specification no longer guarantees efficiency gain in the presence of missing outcomes but \iptw still does. The asymptotic equivalence between the two methods for improving precision therefore breaks down. 

More specifically, Theorem \ref{thm:clt}(i) ensures that adjustment by \iptw reduces the asymptotic variance by $e(1-e)\var \{ \textup{proj}(\tyi \mid 1, x_i)  \}$. This expression does not depend on $p_i$ such that the reduction is the same as the reduction when outcomes are fully observed. 
Observe that adding more covariates to $x_i$ never reduces the variance of $\textup{proj}(\tyi \mid 1, x_i)$. 
Adjusting for more covariates by \iptw hence never hurts the asymptotic efficiency of the resulting $\hteps$. 
This underpins the extension to the case with missingness in both covariates and outcomes in Section \ref{sec:missingxy}.

On the other hand, Theorem \ref{thm:clt}(ii) suggests that $\htreg $ does not ensure efficiency gain over the unadjusted estimator $\htn$ unless the outcomes are missing completely at random or the true outcome model is linear in $x_i$. The latter condition echos the standard result in semiparametric efficiency theory. In particular, \cite{robins2007comment} pointed out that $\htreg$ can be written as a classic augmented inverse propensity score weighted estimator with a linear outcome model that corresponds to the interacted regression;  see Proposition {S2} in the \sm. Standard theory ensures that it achieves semiparametric efficiency if both the missingness model and the outcome model are correctly specified  \citep{tsiatis2006semiparametric}.

\section{Covariate adjustment with missingness in both covariates and outcomes}\label{sec:missingxy}
\subsection{Overview and recommendation}\label{sec:missingxy_overview}
We now extend to the case with missingness in both covariates and outcomes. 
Recall $x_i \in \mathbb R^J$ as the vector of baseline covariates that are fully observed for all units.
Assume that in addition to $x_i$, we also have $K$ partially observed covariates, summarized as $w_i = (w_{i1}, \ldots, w_{iK})\in \mathbb R^K$ for $i = \ot{N}$. 
Of interest is how we may use this additional information to further improve inference.

To this end, we recommend using the missingness indicator method to address missing covariates \citep{zhao2022to} and then constructing the regression-adjusted and propensity-score-weighted estimators based on the augmented covariate vectors from the missingness indicator method; see Algorithm \ref{algo:1} below.  We  show in Section \ref{sec:theory_xy} that the resulting estimators preserve the theoretical properties in Theorem \ref{thm:clt}.
Accordingly, we recommend using regression adjustment when the linear outcome model or missing completely at random assumption is plausible and using propensity score weighting otherwise.
The results combine the theory in Section \ref{sec:missingY} on missing outcomes and that in \cite{zhao2022to} on missing covariates, and offer the full picture of covariate adjustment with missing outcomes and covariates.

Let $R_i^w = (R_{i1}^w, \ldots, R_{iK}^w) \in \{1,0\}^K$ represent the missingness in $w_i$, with $R_{ik}^w = 1$ if $w_{ik}$ is observed and $R_{ik}^w = 0$ if $w_{ik}$ is missing. 
Let $\wiz \in \mathbb R^K$ denote an imputed variant of $w_i$, where we impute all missing elements with zero. Note that $\wiz$ is in fact the elementwise product, or intuitively the ``interaction", between $w_i$ and  $R^w_i$ regardless of the actual values of the missing elements  in $w_i$. 
The concatenation of $(x_i, \wiz, \riw)$, denoted by $x_i^\aug    \in \mathbb R^{J+2K}$, gives the vector of fully observed covariates under the missingness indicator method, which imputes all missing covariates by zero and augments the completed covariates by the corresponding missingness indicators. 
We use the superscript ``mim'' to signify the missingness indicator method.

Observe that $x_i^\mim$ summarizes all observed information in $(x_i,w_i)$. 
Renew $e_i = \pr(Z_i = 1\mid x_i^\mim) = e$ and $p_i = \pr(\ryi = 1 \mid x_i^\mim, Z_i)$  as the propensity score and the probability of having an observed outcome given $x_i^\mim$. 
Algorithm \ref{algo:1} below states the procedure for constructing the recommended estimators, denoted by $\htreg(x_i^\mim)$ and $\hteps(x_i^\mim)$, as variants of $\htreg$ and $\hteps$ after replacing $x_i$ with  $x_i^\mim$ as the new fully observed covariate vector; c.f$.$ Table \ref{tb:estimators}.
 
\begin{algo}\label{algo:1}
\begine[(i)]
\item Construct $x_i^\aug = (x_i, \wiz, \riw)$ as the new fully observed covariate vector;
\item\label{item:step_p}  Estimate $p_i$ from the prespecified outcome missingness model, denoted by $\hat p_i$. When the outcome missingness model is unknown, compute $\hat p_i$ from the logistic regression $\glmt(\ryi \sim  1 + x_i^\aug + Z_i + x_i^\aug Z_i)$ over $i = \ot{N}$;
\item\label{item:step_reg}    When the linear outcome model or missing completely at random assumption is plausible, compute $\htreg(x_i^\mim)$ as the coefficient of $Z_i$ from the \wls fit of the interacted regression $
\lmt \{Y_i \sim 1 + Z_i + (x_i^\mim - \bar x^\mim) + Z_i(x_i^\mim - \bar x^\mim) \}$ over $\{i: \riy = 1\}$, where we weight unit $i$ by $\hat p_i^{-1}$.

\smallskip 
 
\item[] Otherwise, estimate $e_i$ from the logistic regression $\glmt(Z_i \sim 1+x_i^\mim)$ over $i = \ot{N}$, denoted by $\hat e_i$, and compute $\hteps (x_i^\mim)$ as the coefficient of $Z_i$ from the \wls fit of the unadjusted regression $\lmt(Y_i \sim 1+Z_i)$ over $\{i: \riy = 1\}$, where we weight unit $i$ by $\hat p_i^{-1} \{Z_i / \hat e_i +(1-Z_i)/(1-\hei)\}$.

\ende
\end{algo}

Despite the apparent oversimplification by imputing all missing covariates with zero in forming $ x_i^\mim $, the resulting $\htreg(x_i^\mim)$ and $\hteps(x_i^\mim)$ are invariant over a general class of imputation schemes. Specifically, consider a covariate-wise imputation strategy where  for $k = \ot{K}$, we impute all missing values in the $k$th partially observed covariate by a common value $c_k\in\mathbb R$ \citep{zhao2022to}.
Let $c = (c_1, \ldots, c_K)$ represent the imputation scheme. The resulting imputed variant of $w_i$ equals $\wic  = (w_{i1}^{c}, \ldots, w_{iK}^{c})^\T\in \mathbb R^K$ with $w_{ik}^{c} = w_{ik}$ if $w_{ik}$ is observed and $w_{ik}^{c} = c_k$ otherwise. 
This defines a general class of imputed variants of $w_i$ that includes $\wiz$ as a special case with $c_k = 0$ for all $k$. Another common choice of $c_k$ is the average of the observed values in $(w_{ik})_{i=1}^N$. 
Let $x_i^\mim(c) = (x_i, \wic, \riw)\in \mathbb R^{J + 2K}$ denote a variant of $x_i^\mim$ where we use the more general $\wic$ in place of $\wiz$. 

\begin{proposition}\label{prop:invariance} 
Assume that in Algorithm \ref{algo:1}, we replace all $x_i^\mim$ by $x_i^\mim(c)$. The resulting estimators are invariant to the choice of the imputed values and equal  $\htreg(x_i^\mim)$ and $\hteps(x_i^\mim)$ for all $c \in \mathbb R^K$.
\end{proposition}

\subsection{Asymptotic justification of the recommended estimators in Algorithm \ref{algo:1}}\label{sec:theory_xy}

Theorem \ref{thm:ipwx_clt_xy} below states the asymptotic properties of $\htreg (x_i^\mim)$ and $\hteps (x_i^\mim)$ in Algorithm \ref{algo:1}. 
For comparison, let $x_i^\sub$ be a subvector of $x_i^\mim$ and  let $\hteps (x_i^\sub)$ denote a variant  of $\hteps (x_i^\aug)$ with  $x_i^\mim$ replaced by $x_i^\sub$ in step \eqref{item:step_reg} of Algorithm \ref{algo:1}.
A common choice of $x_i^\sub$ is $x_i^\sub = x_i$, where we use only fully observed covariates in constructing $\hat e_i$. 

\begin{theorem}\label{thm:ipwx_clt_xy}
Assume complete randomization with $Z_i \indep \{  Y_i(1), Y_i(0), x_i, w_i,  R_i^w\}$ and Assumption \ref{assm:riy} holds with all $x_i$ replaced by $x_i^\mim$. Under standard regularity conditions as $N\to\infty$, we have
\begine[(i)]
\item Theorem \ref{thm:clt} holds with $(\htreg, \hteps)$ replaced by $(\htreg(x_i^\mim), \hteps(x_i^\mim))$ and $\htn$ renewed based on the renewed definition of $\hpi$ from step \eqref{item:step_p} of Algorithm \ref{algo:1};
\item the asymptotic variance of $\hteps(x_i^\sub)$ is greater than or equal to that of $\hteps(x_i^\mim)$. 
\ende
\end{theorem}

Other than being independent of $Z_i$, Theorem \ref{thm:ipwx_clt_xy} does not require further assumptions on the missingness mechanism of $w_i$. Therefore, Theorem \ref{thm:ipwx_clt_xy} holds even if $w_i$ is {\it missing not at random}, which departs from the standard literature of missing covariates under the {\it missing-at-random} assumption \citep{robins1994estimation}.

In addition, the independence between $\riw$ and $Z_i$ ensures $x_i^\mim = (x_i, \wiz, \riw)$ is effectively a fully observed pretreatment covariate vector, generalizing $x_i$. Theorem \ref{thm:ipwx_clt_xy} requires that Assumption \ref{assm:riy} holds with all $x_i$ replaced by $x_i^\mim$, and thereby ensures the outcome is missing at random with $\riy$ $\indep$ $Y_i \mid (x_i^\mim, Z_i)$. This is the weakest form of the missing-at-random assumption based on the observed information in $(x_i, w_i, Z_i)$ \citep{rosenbaum1984reducing}.  
An alternative, more standard form of the missing-at-random assumption is 
$\ryi \indep Y_i \mid (x_i, Z_i)$. This is a more restrictive condition because it does not allow the missingness in outcomes to depend on the missingness pattern of the covariates, as represented by $\riw$, or the observed values in $w_i$.

Therorem \ref{thm:ipwx_clt_xy}(i) ensures that  $\hteps (x_i^\aug)$ is asymptotically more efficient than $\htn $ while $\htreg(x_i^\mim)$ is asymptotically more efficient than both $\htn$ and $\hteps(x_i^\mim)$ when the outcomes are either linear in covariates or missing completely at random. Therorem \ref{thm:ipwx_clt_xy}(ii) ensures that $\hteps (x_i^\aug)$ is asymptotically more efficient than all alternative propensity-score-weighted estimators  that use only a subset of $x_i^\mim$ for estimating the propensity score, including $\hteps(x_i)$ that uses only $x_i$. Recall that $\wiz$ is a variant of $w_i$ by imputing all missing covariates by zero. 
From Theorem \ref{thm:ipwx_clt_xy}, this rather basic imputation guarantees efficiency gain irrespective of the true values of the missing covariates, illustrating the quick wins that can be achieved by adjusting for partially observed covariates.

As a comparison, the approach reviewed by \cite{seaman2013review} applies inverse probability weighting to only units with fully observed outcome {\it and} covariates, and requires correct specification of both the outcome and covariate missingness models. Accordingly, it requires the covariates to be missing at random. Our proposed method, in contrast, places no restriction on the covariate missingness mechanism and does not require the specification of the covariate missingness model.

\section{Simulation and a real-data example}\label{sec:simulation}

\subsection{Regression adjustment does not guarantee efficiency gain}

We first illustrate the possibly worse precision of $\htreg $ in finite samples. 
Assume missingness in only outcomes. We generate $\{x_i, Y_i(1), Y_i(0), Z_i, \ryi\}_{i=1}^N$ as independent realizations of 
$
x_i \sim \textup{Uniform}(-10,10)$, $Y_i(1) = \sin(x_i)$, $Y_i(0) = -\cos(x_i)$, $Z_i \sim \text{Bernoulli}(e)$, and $\ryi \sim \text{Bernoulli}(p_i)$, where 
$
p_i = \{1+\exp(-1-2x_i)\}^{-1} 
$. 
Figure \ref{fig:simu}(a) shows the distributions of the deviations of $\{\htn  ,\htreg  ,\hteps   \}$ from $\tau$ over 10,000 independent replications at $N=1$,000 and $e=0.2, 0.5$. 
The regression-adjusted $\htreg$ has worse precision than $\htn$ in both cases. 
Similar patterns are observed for other choices of potential outcomes and combinations of $N$ and $e$; we omit the results to avoid redundancy. 

\begin{figure}[t]
\begin{center}
\begin{tabular}{cccccc}
\includegraphics[width=.3 \textwidth]{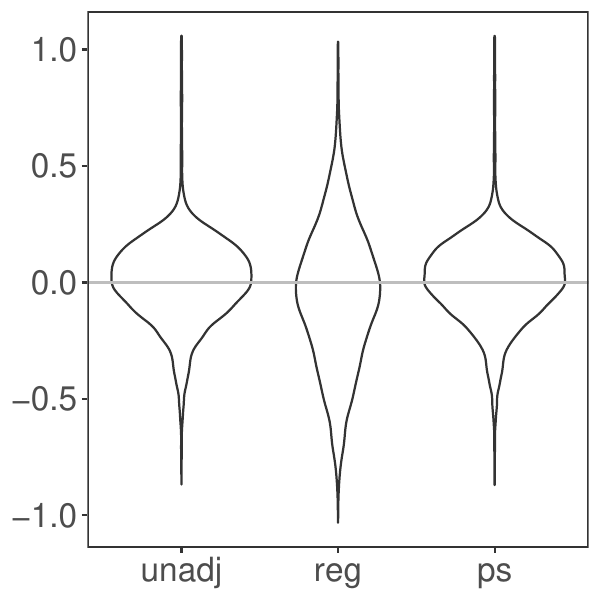} & \includegraphics[width=.3 \textwidth]{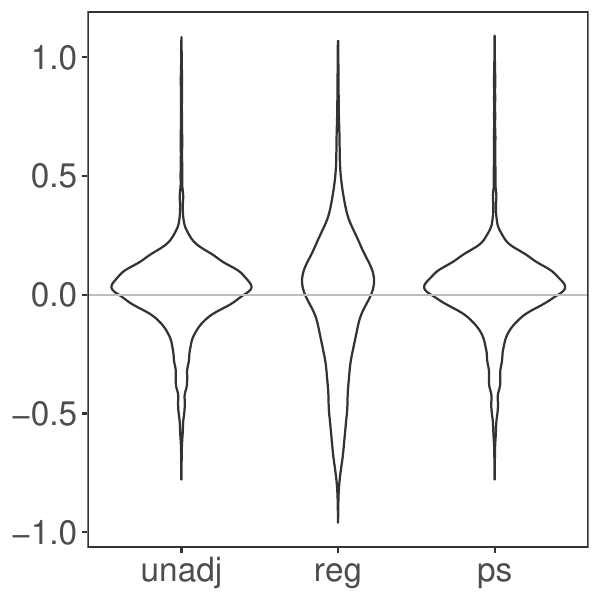} & \includegraphics[width=.3 \textwidth]{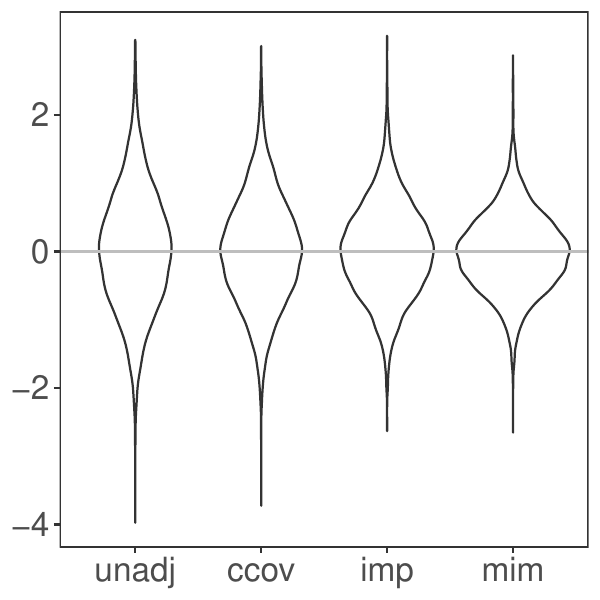}
\\
\footnotesize  $e=0.2$ & \footnotesize  $e=0.5$ & \\
\multicolumn{2}{c}{\footnotesize 
\begin{tabular}{c}
(a)  Deviations of $\{\htn ,\htreg ,\hteps \}$
 from $\tau$.
\end{tabular}} & \footnotesize  
\begin{tabular}{c}
(b) Deviations of
 $\htn$ (unadj), \\
 $\hteps(x_i)$ (ccov), 
 $\hteps(x_i^\imp)$ (imp), \\
 and $\hteps(x_i^\mim)$ (mim) 
from $\tau$.
\end{tabular}
\end{tabular}
\end{center}
\caption{\label{fig:simu} Violin plots over $10^4$ independent replications.}
\end{figure}

\subsection{Efficiency gain by \iptwn}\label{sec:simu_ipwx}
We next illustrate the efficiency gain by \iptwn, along with the benefits of adjusting for partially observed covariates. 
Consider a treatment-control experiment with $N = 500$ units and treatment probability $e = 0.2$. 
Let $(\xi_i)_{i=1}^N$ be independent Bernoulli$(0.4)$ to divide the units into two latent classes.
Assume that we have $J = 1$ fully observed covariate $x_i \sim \mn(\xi_i,1)$ and $K=9$ partially observed covariates $w_i =  (w_{i1},   \ldots, w_{iK})^\T   \sim \mn( \xi_i 1_{K}, I_K)$.
We 
generate the missingness indicators and potential outcomes as
$R^w_{ik}  \sim \textup{Bernoulli}\{0.5\xi_i + 0.95(1-\xi_i)\}$ for $k = \ot{K}$, $\ryi \sim \text{Bernoulli}(p_i)$, 
where $p_i = \{1+\exp(-1-2x_i)\}^{-1}$, and $Y_i(z) \sim \mathcal{N}\{\mu_i(z), 1\}$, where $\mu_i(z) = 3\xi_i  + (x_i+ \sum_{k=1}^K w_{ik})   \gamma_{z | \xi_i} + 3\sum_{k=1}^K R^w_{ik}$  with 
$(\gamma_{1|1}, \gamma_{0|1}) =(1, -1) $  and $(\gamma_{1|0}, \gamma_{0|0}) =(0.5, - 0.5)$. The data generating process ensures the covariates are missing not at random with units with $\xi_i = 1$ having both a higher chance of missing covariates and on average greater values of covariates.

Let $\hteps(x_i)$ denote a variant of $\hteps(x_i^\mim)$ where we use only the fully observed $x_i$ in constructing $\hei $. 
Let $\hteps(x_i^\imp)$ denote a variant of $\hteps(x_i^\mim)$ where we use only the union of $x_i$ and $\wiz$ in constructing $\hei $.
Figure \ref{fig:simu}(b) shows the distributions of the deviations of $\htn$, $\hteps(x_i) $, $\hteps(x_i^\imp)$, and $\hteps(x_i^\mim) $ from $\tau$  over 10,000 independent replications.
The results are coherent with the asymptotic theory in Theorem \ref{thm:ipwx_clt_xy}, with $\hteps(x_i^\mim)$ being the most precise.

\subsection{A real-data example}
We now apply the proposed method to the Best Apnea Interventions for Research trial in \cite{bakker2016motivational}. 
A total of 169 patients were recruited and randomized with equal probability to active treatment and control. 
One outcome of interest is the 24-hour systolic blood pressure  measured at six months, which is missing for 45 patients.

As an illustration, we consider four baseline covariates, namely age, gender, baseline Apnea-Hypopnea Index, and baseline diastolic blood pressure, for estimating the outcome missingness model and covariate adjustment.  
The first three covariates are fully observed, and the last covariate is missing for eight patients. 
Table \ref{tb:bestAir} summarizes the point estimates and estimated variances of the unadjusted, regression-adjusted, and propensity-score-weighted estimators based on the fully observed covariates and the augmented covariates under the missingness indicator method,  respectively.
The variances are estimated by the bootstrap over 10,000 independent replications.  
The results are coherent with the asymptotic theory, with the combination of \iptw and missingness indicator method ($\hteps(x_i^\mim)$) resulting in the smallest bootstrap variance. 
The two regression-adjusted estimators $\htreg(x_i^\mim)$ and $\htreg(x_i)$, on the other hand, have higher bootstrap variances than their respective unadjusted counterparts $\htn(x_i^\mim)$ and $\htn(x_i)$, illustrating the possible loss in precision by regression adjustment.

\begin{table}[t]
\caption{\label{tb:bestAir} Point estimates and estimated variances of the unadjusted, regression-adjusted, and propensity-score-weighted estimators based on $x_i$ and $x_i^\mim$, respectively. The variances are estimated by the bootstrap over $10^4$ independent replications.}
\begin{center}  
\begin{tabular}{l|ccc|ccc}\hline
 & $\htn(x_i^\mim)$ & $\htreg(x_i^\mim)$ & $\hteps(x_i^\mim)$ & $\htn(x_i)$ & $\htreg(x_i)$ & $\hteps(x_i)$ \\\hline
Point estimate & $-7.03$ & $-6.07$  & $-5.47$  & $-4.78$ & $-4.66$  & $-4.60$ \\ 
Estimated variance 
 & 6.26 &   7.50 &   4.72 &   5.69 &  5.85 &  5.63 \\\hline
\end{tabular}
\end{center}

\end{table}

\vspace{-1mm}

\section{Further discussion on the role of the outcome model}\label{sec:ext}

Theorems \ref{thm:clt} and \ref{thm:ipwx_clt_xy} assume that the missingness model for the outcome is correctly specified. When this assumption fails, $\hteps$ is inconsistent, while $\htreg$ remains consistent if the linear outcome model is correct. The use of the outcome model ensures this double robustness property of $\htreg$; see Proposition {S2} in the \sm. Analogously, we can also augment $\hteps$ with the outcome model:
\begina
\htau_{\textup{ps-reg}} &=& \frac{1}{N}\sum_{i=1}^N \left[   \hat m_1(x_i) + \dfrac{\riy}{\hat p_i} 
\dfrac{Z_i}{\hei }  \left\{ Y_i - \hat m_1(x_i) \right\} \right] -  \frac{1}{N}\sum_{i=1}^N\left[    \hat m_0(x_i) + \dfrac{\riy}{\hpi} \dfrac{1-Z_i}{1-\hei} \left\{ Y_i - \hat m_0(x_i)\right\} \right],
\enda
where $\hat m_z(x_i) \ (z=0,1)$ is the estimated outcome model. The augmented estimator $\htau_{\textup{ps-reg}}$ is doubly robust in that it is consistent if either the outcome model or the outcome missingness model is correct.
As a special case, we can construct $\htau_{\textup{ps-reg}}$ as the coefficient of $Z_i$ from the \wls fit of the interacted regression over $\{i: \riy = 1\}$ with weight $\hpi^{-1}\hat \pi_i $ for unit $i$. The corresponding $\hat m_z(x_i) \ (z=0,1)$ equals the estimated outcome model from the same weighted-least-squares fit; see Proposition S3 in the Supplementary Material.   This integrates the regression adjustment and the \iptw in the last two rows of Table \ref{tb:estimators}.

 \bibliographystyle{chicago}
\bibliography{refs_missingXY}

\newpage
\setcounter{equation}{0}
\setcounter{section}{0}
\setcounter{figure}{0}
\setcounter{example}{0}
\setcounter{proposition}{0}
\setcounter{corollary}{0}
\setcounter{theorem}{0}
\setcounter{table}{0}
\setcounter{condition}{0}
\setcounter{lemma}{0}
\setcounter{remark}{0}

\renewcommand {\theproposition} {S\arabic{proposition}}
\renewcommand {\theexample} {S\arabic{example}}
\renewcommand {\thefigure} {S\arabic{figure}}
\renewcommand {\thetable} {S\arabic{table}}
\renewcommand {\theequation} {S\arabic{equation}}
\renewcommand {\thelemma} {S\arabic{lemma}}
\renewcommand {\thesection} {S\arabic{section}}
\renewcommand {\thetheorem} {S\arabic{theorem}}
\renewcommand {\thecorollary} {S\arabic{corollary}}
\renewcommand {\thecondition} {S\arabic{condition}}
\renewcommand {\thepage} {S\arabic{page}}

\setcounter{page}{1}

\onehalfspacing

\begin{center}
\bf \Large 
Supplementary material for ``Covariate adjustment in randomized experiments with missing outcomes and covariates"
\end{center}

Section \ref{sec:additional results} gives the additional results that underlie the comments in the main paper. In particular, Theorem \ref{thm:clt_app} gives the complete version of Theorem \ref{thm:clt}. 

Section \ref{sec:lemma} gives the lemmas.

Section \ref{sec:proof_thm1} gives the proof of Theorem \ref{thm:clt_app}. 

Section \ref{sec:proof_other} gives the proofs of the rest of the results in the main paper and Section \ref{sec:additional results}. 

\bigskip 

Denote by  
 \beginy\label{eq:p_logit}
\pfun = \dfrac{1}{1+\exp(-U_i^\T \beta)} 
\endy the logistic outcome missingness model for $p_i = \pr(\riy = 1\mid x_i, Z_i)$ under Assumption \ref{assm:riy}. 
Recall $\beta^*$ as the true value of the parameter with $p_i = p(x_i, Z_i; \beta^*)$.
Let $\hat\beta$ denote the maximum likelihood estimate of $\beta^*$ with $\hat p_i = p(x_i, Z_i; \hat\beta)$.

Recall $\hei $ as the estimate of the propensity score $e_i = \pr(Z_i = 1\mid x_i)$ from the logistic regression $\glmt(Z_i \sim 1+x_i)$.
The underlying logistic model is 
\beginy\label{eq:e_logit}
 e(x_i; \alpha) = \dfrac{1}{1+\exp(-\txi^\T \alpha)},
\endy
where $\txi = (1, x_i^\T)^\T$, and we have $\hei = e(x_i; \hat\alpha)$, where $\hat\alpha$ denotes the maximum likelihood estimate of $\alpha$ under \eqref{eq:e_logit}.
Under Bernoulli randomization, model \eqref{eq:e_logit} is for sure correctly specified  with the true value of $\alpha$ equal to $\alpha^* = ( \log\{e/(1-e)\}, 0_J^\T)^\T$  with $e(x_i; \alpha^*)  =e$.

Let $\rs$ denote convergence in distribution. Let $\va(\cdot)$ and $\cova(\cdot)$ denote the asymptotic variance and covariance.
All asymptotic results in the main paper and this \sms follow from standard theory of m-estimation; see  \cite{newey1994large}. We assume standard regularity conditions for m-estimation throughout. 

Table \ref{tb:estimators_app} extends Table \ref{tb:estimators} in the main paper and summarizes the four estimators we discussed in Sections 2 and 5 of the main paper.  
\begin{table}[h]\caption{\label{tb:estimators_app} 
Let $x_i' = x_i - \bar x$ denote the centered covariates, where $\bar x = \meani x_i$. Let $\hat p_i$ denote the estimated probability of $Y_i$ being observed given $(x_i, Z_i)$, and let $\hat \pi_i = Z_i / \hat e_i  + (1-Z_i) / (1-\hat e_i)$ denote the inverse of the estimated probability of the treatment received for unit $i$ with 
$\hat e_i$ as the estimated propensity score. 
}

 \begin{center}
 \begin{tabular}{l|c|c }\hline
 & Regression specification & Weight over $\{i:\riy = 1\}$ \\\hline
 $\htn $ & $\lmt(Y_i \sim 1+Z_i)$  &  $\hpi^{-1}$ \\\hline
 $\htreg $ & $\lmt(Y_i \sim 1+Z_i + x_i' + Z_i x_i')$  &   $\hpi^{-1}$ \\\hline
 $\hteps $ & $\lmt(Y_i \sim 1+Z_i)$  
 &  
 $\hpi^{-1} \hat \pi_i $\\\hline
 $\htau_\textup{x-ps-reg} $ & $\lmt(Y_i \sim 1+Z_i + x_i' + Z_i x_i')$  &  $\hpi^{-1} \hat \pi_i $\\\hline
 
 \end{tabular}
 \end{center}
 
 \end{table}

\section{Additional results}\label{sec:additional results}

\subsection{Details of Theorem \ref{thm:clt}} 

Recall $U_i$ as the covariate vector for the logistic outcome missingness model under Assumption \ref{assm:riy}. Let $\muo  = E\{Y_i(1)\}$, $\muz  = E\{Y_i(0)\}$, $\mu_x = E(x_i)$, $\go = \covxinv \cov\{x_i, \yio\}$ and $\gz  = \covxinv \cov\{ x_i, \yiz\}$ be shorthand notation for the population moments; let 
\begina
\begin{array}{lll}
b_{11} = 
 E\left[ p_i^{-1}  \cdot Z_i \cdot  \dyon ^2   \right],& 
 \quad  b_{22} = 
 E\left[ p_i^{-1}  \cdot (1-Z_i) \cdot  \dyzn  ^2   \right],\\
b_{13} 
  = E  \left[ (1-p_i) \cdot Z_i \cdot \dyon  \cdot U_i^\T  \right],& \quad
b_{23} 
   = E  \left[ (1-p_i) \cdot (1-Z_i) \cdot \dyzn  \cdot U_i^\T  \right],\\
 b_{33} = E\left\{p_i(1-p_i) \cdot U_i U_i^\T \right\},
\end{array}
\enda
and let $(b_{11}', b_{22}', b_{13}', b_{23}')$  be variants of $(b_{11}, b_{22},  b_{13}, b_{23})$ by replacing  $Y_i(1)$, $\mu_1 = E\{\yio\}$, $Y_i(0)$, $\mu_0=E\{\yiz\}$ with their respective covariate-adjusted analogs 
$Y_i'(1) = Y_i(1)-x_i^\T\go$, $\mu_1' = E\{Y_i'(1)\} = \muo  - \mu_x^\T\go$, $Y_i'(0) = Y_i(0)-x_i^\T\gz$ and $\mu_0' = E\{Y_i'(0)\} = \muz  - \mu_x^\T\gz$.
Theorem \ref{thm:clt_app} below underpins Theorem \ref{thm:clt} and gives the explicit forms of the asymptotic variances. 

\begin{theorem}\label{thm:clt_app}
Assume Assumption \ref{assm:riy} and  standard regularity conditions as $N\to\infty$. We have 
$
\sqrt N(\hat\tau_\dmd - \tau) \rs  \mn\left(0,  v_\dmd \right)$ for $\dmd = \ua,  \wlsx, \hajx$, 
where 
\begina
v_\ua &=& \dfrac{b_{11}}{e^2} + \dfrac{b_{22}}{(1-e)^2} - \left(\dfrac{b_{13}}{e}-\dfrac{b_{23}}{1-e}\right)b_{33}^{-1}\left(\dfrac{b_{13}}{e}-\dfrac{b_{23}}{1-e}\right)^\T,\\
v_\xreg &=&  \dfrac{b'_{11}}{e^2} + \dfrac{b'_{22}}{(1-e)^2} - \left(\dfrac{b'_{13}}{e}-\dfrac{b'_{23}}{1-e}\right)b_{33}^{-1}\left(\dfrac{b'_{13}}{e}-\dfrac{b'_{23}}{1-e}\right)^\T+(\go-\gz)^\T\cov(x_i)(\go-\gz),\\
v_\ipwx &=& v_\ua - e(1-e)\var \{ \textup{proj}(\tyi \mid 1, x_i)  \}.
\enda 
\begine[(i)] 
\item\label{item:complete}  With fully observed data,  that is,  $p_i$ = 1 for all $i$, we have 
\begina
v_\ua = \dfrac{ \var\{Y_i(1)\}}{e } + \dfrac{\var\{Y_i(0)\}}{1-e}, \quad 
v_\xreg = v_\ipwx = v_\ua - e(1-e) \var\{\textup{proj}(\tyi \mid 1, x_i)\}. 
\enda
%
%
\item\label{item:mcar} 
 When the outcome data are missing completely at random with $p_i = p \in (0,1)$ for all $i$, we have 
\begina
v_\ua &=& \dfrac{\var\{Y_i(1)\}}{ep} + \dfrac{\var\{Y_i(0)\}}{(1-e)p},\\
v_\xreg &=&  \dfrac{\var\{Y_i(1)\}  - \gamma_1^\T\cov(x_i) \gamma_1}{ep} + \dfrac{\var\{Y_i(0)\} - \gamma_0^\T\cov(x_i) \gamma_0}{(1-e)p}+(\go-\gz)^\T\cov(x_i)(\go-\gz)
\enda
with 
\begina
v_\xreg - v_\xps =(1-p^{-1})\left\{ \frac{\gamma_1^\T \cov(x_i) \gamma_1}{e}   + \frac{\gamma_0^\T \cov(x_i) \gamma_0}{1-e}  \right\} \leq 0.
\enda

\item\label{item:linear} 
When the outcome model is linear in $x_i$ with $E\{Y_i(z) \mid x_i\} = \muzz  + (x_i - \mu_x)^\T \gamma_z$ for $z = 0,1$,
we have
\begina
v_\ua - v_\xreg 
=  \va(\htn  - \htreg ) 
  \geq  0, \qquad  v_\xps - v_\xreg = E(\Gamma^\T\Gamma) \geq 0,
\enda
where $\Gamma  = B - A\{E(A^\T A)\}^{-1}E(A^\T B)$ is the residual of the linear projection of $B$ on $A$ for \begina
A = \sqrt{p_i(1-p_i)} U_i^\T, \quad B = \sqrt{\dfrac{1-p_i}{p_i}} (x_i-\mu_x)^\T  \left( \frac{Z_i}{e}  \gamma_1- \frac{1-Z_i}{1-e} \gamma_0\right).   
\enda
\ende 

\end{theorem}

We give the proof of Theorem \ref{thm:clt_app} in Section \ref{sec:proof_thm1}. 
To conduct inference, we can estimate the asymptotic variances by the standard empirical sandwich variance estimator for m-estimation \citep{newey1994large} or the bootstrap by resampling units $i=1, \ldots, N$ with replacement.

\subsection{\hajek forms of $\htn$ and $\hteps$}
Proposition \ref{prop:hajek_app} below shows that the unadjusted and propensity-score-weighted estimators are both \hajek estimators of $\tau = E\{\yio\} - E\{\yiz\}$ under different weighting schemes.
The result generalizes the observations by  \cite{imbens2004nonparametric} and \citet[][Theorem 14.2]{ding2023a} to the presence of missing data.
\begin{proposition}\label{prop:hajek_app}

\begine[(i)]
\item $\htn = \hyua(1) - \hyua(0)$, where 
\begina
\hyua(1) = \dfrac{\sumi  \dfrac{\riy}{\hpi} \cdot Z_i \cdot Y_i }{\sumi  \dfrac{\riy}{\hpi}\cdot Z_i}, \quad \hyua(0) = \dfrac{\sumi  \dfrac{\riy}{\hpi}  \cdot(1-Z_i) \cdot  Y_i}{\sumi   \dfrac{\riy}{\hpi}\cdot(1-Z_i)}.
\enda 
\item $
\hteps = \hmuoipwx  - \hmuzipwx$, where
\begina
\hmuoipwx  = \dfrac{\sumi \dfrac{\riy}{\hat p_i}
\cdot 
\dfrac{Z_i}{\hei } \cdot Y_i }{\sumi \dfrac{\riy}{\hat p_i}
\cdot 
\dfrac{Z_i}{\hei }} ,\quad
\hmuzipwx  = \dfrac{\sumi\dfrac{\riy}{\hpi} \cdot \dfrac{1-Z_i}{1-\hei}\cdot  Y_i}{\sumi \dfrac{\riy}{\hpi} \cdot \dfrac{1-Z_i}{1-\hei}}.
\enda 
\ende 
The $\htn$ and $\hteps$ when all data are observed are special cases with $\ryi = 1$ and $\hpi = 1$ for all $i$.
\end{proposition}

\subsection{Double robustness of $\htreg$ and $\htau_\textup{x-ps-reg}$}
For $z = 0,1$, let  $\hat m_z(x_i)$ denote the estimate of  $E(Y_i \mid x_i, Z_i = z) $ based on the weighted-least-squares fit of the interacted regression $\lmt\{Y_i \sim 1 + Z_i + (x_i - \bar x) + Z_i(x_i - \bar x)\}$, where we weight unit $i$ by $\hpi^{-1}$. 
From Proposition \ref{prop:hajek_app}, we can also view $\htn$ as a  \hajek  estimator, where we weight $Y_i$ by $\dfrac{ \riy}{ \hat p_i } \cdot \dfrac{Z_i  }{e  }$ in estimating $E\{\yio\}$ and by $\dfrac{\riy}{\hat p_i} \cdot \dfrac{1-Z_i}{1-e}$ in estimating $E\{\yiz\}$.
Proposition \ref{prop:wlsx_dr} below expresses $\htreg$ as a classic augmented inverse propensity score weighted estimator based on $\hat m_z(x_i)$'s and the weighting scheme of $\htn$ \citep{robins2007comment, ding2023a}. 
\begin{proposition}\label{prop:wlsx_dr} 
$ \htreg $ is doubly robust with respect to outcome models $\{\hmo(x_i), \hmz(x_i)\}$ and the weighting scheme of $\htn $. That is, 
$
\htreg  = \hydro - \hydrz$, 
where
\begina
\hydro &=& \meani \left[ \hat m_1(x_i)   + \dfrac{ \riy}{ \hat p_i } \cdot \dfrac{Z_i  }{e  } \cdot \left\{ Y_i - \hat m_1(x_i)  \right\}\right], \nonumber \\
\hydrz &=&  \meani \left[ \hat m_0(x_i)   + \dfrac{\riy}{\hat p_i} \cdot \dfrac{1-Z_i}{1-e}\cdot \left\{ Y_i - \hat m_0(x_i)  \right\}\right].
\enda
\end{proposition}

Proposition \ref{prop:reg-ps} below parallels Proposition \ref{prop:wlsx_dr} and formalizes the double-robustness of $\htau_\textup{x-ps-reg}$ in Table \ref{tb:estimators_app}. 
For $z = 0,1$, let  $\hat m'_z(x_i)$ denote the estimate of  $E(Y_i \mid x_i, Z_i = z) $ based on the weighted-least-squares fit of the interacted regression $\lmt\{Y_i \sim 1 + Z_i + (x_i - \bar x) + Z_i(x_i - \bar x)\}$, where we weight unit $i$ by $\hpi^{-1}\hat \pi_i$. 

\begin{proposition}\label{prop:reg-ps}
$\htau_\textup{x-ps-reg}$ is doubly robust with respect to outcome models $\{\hmo'(x_i), \hmz'(x_i)\}$ and the weighting scheme of $\hteps$. That is, 
$
\htau_\textup{x-ps-reg}  = \hat Y'_\textup{dr}(1) - \hat Y'_\textup{dr}(0)$, 
where
\begina
\hat Y'_\textup{dr}(1) &=& \meani \left[ \hat m_1'(x_i)   + \dfrac{ \riy}{ \hat p_i } \cdot \dfrac{Z_i  }{\hat e_i} \cdot \left\{ Y_i - \hat m_1'(x_i)  \right\}\right], \nonumber \\
\hat Y'_\textup{dr}(0) &=&  \meani \left[ \hat m_0'(x_i)   + \dfrac{\riy}{\hat p_i} \cdot \dfrac{1-Z_i}{1-\hat e_i}\cdot \left\{ Y_i - \hat m_0'(x_i)  \right\}\right].
\enda
\end{proposition}

\subsection{Nondecreasing efficiency of covariate adjustment by \iptw}
Proposition \ref{prop:eps_eff} below formalizes the intuition that adjusting for more covariates by \iptw never harms the asymptotic efficiency. 

\begin{proposition}\label{prop:eps_eff}
Assume missingness in only outcomes with the outcome missingness mechanism given by Assumption \ref{assm:riy}. Let $\htau_{\ipwx}'$ be a variant of $\hteps$, where we construct the estimated propensity score using $\glmt(Z_i \sim 1 +x'_i)$ over $i = \ot{N}$ for $x_i'$ that is distinct from $x_i$.  
Let $X = (x_1, \ldots, x_N)^\T$ and $X' = (x'_1, \ldots,x'_N)^\T$ denote the matrices with $x_i$ and $x'_i$ as the $i$th row vectors, respectively. 
If the column space of $X'$ is a subspace of that of $X$,
then $\hteps$ is asymptotically at least as efficient as $\hteps'$ with $ \va(\hteps) \leq  \va(\hteps')$. 
A special case for the condition is when $x_i'$ is a subset of $x_i$. 
\end{proposition}

\section{Lemmas}\label{sec:lemma}
Recall $p(x_i, Z_i; \beta)$ as the logistic outcome missingness model under Assumption \ref{assm:riy}, with the true value of the parameter denoted by $\beta^*$. 
Recall $\hat\beta$ as the maximum likelihood estimate of $\beta^*$.

\begin{lemma}\label{lem:beta}
Under Assumption \ref{assm:riy}, 
\begine[(i)]
\item \beginy\label{eq:pdb}
\pd{\beta} p_i = \left.\pd{ \beta } \pfun \right|_{\beta = \beta^*} = p_i(1-p_i)U_i, \qquad 
\pd{\beta} p_i^{-1} = \left.\pd{ \beta } \pfun^{-1} \right|_{\beta = \beta^*} = - \dfrac{1-p_i}{p_i} U_i,
\endy
where we use $\pd{\beta} p_i$ and $\pd{\beta} p_i^{-1}$ as  shorthand notation for the derivatives of $\pfun$ and $\pfun^{-1}$ evaluated at the true value $\beta^*$;  
\item 
the maximum likelihood estimate $\hat\beta$ solves 
$
0  =  \sum_{i=1}^N \psi_\beta\left(x_i, Z_i, \ryi; \beta\right)$, 
where 
\begina
\psi_\beta\left(x_i, Z_i, \ryi; \beta\right) = U_i \left\{\riy - \pfun\right\} 
\enda
with 
\begina
\begin{array}{clllll}
\psi_\beta^* &=& \psi_\beta\left(x_i, Z_i, \ryi; \beta^*\right) &=& U_i \big(\ryi - p_i\big),\\
  \left.  \pd{\beta^\T} \psi_\beta\right|_{\beta = \beta^*} &=&  - p_i(1-p_i) \cdot U_i U_i^\T ,\\
   E(\psi_\beta^* \psi_\beta^{*\T}) &=& - E \left( \left.  \pd{\beta^\T} \psi_\beta\right|_{\beta = \beta^*} \right) &=&E\left\{p_i(1-p_i) \cdot U_i U_i^\T \right\}.
   \end{array}
  \enda  
 \ende
\end{lemma}

The proof of Lemma \ref{lem:beta} follows from standard results for logistic regression.  We omit the details.

\bigskip 

\begin{lemma}\citep{newey1994large, lavergne2008cauchy} \label{lem:cs} Let $A\in \mathbb R^{n\times p}$ and $B \in \mathbb R^{n\times q}$ be random matrices such that $E\|A\|^2 <\infty$, $E\|B\|^2<\infty$, and $E(A^\T A)$ is non-singular. Then 
\begina
E(B^\T B) - E(B^\T A) \{E(A^\T A)\}^{-1} E(A^\T B) = E( \Gamma^\T \Gamma ) \geq 0,
\enda
where $\Gamma  = B - A\{E(A^\T A)\}^{-1}E(A^\T B)$ is the residual of the linear projection of $B$ on $A$.  
\end{lemma}

Lemma \ref{lem:cs} can be verified by direct algebra.  We omit the details.

\bigskip

For $z = 0, 1$, let $\hat \gamma_z$ denote the coefficient vector of $x_i$ from the weighted-least-squares fit of $\lmt(Y_i \sim 1+x_i)$ over units with $Z_i = z$ and $\riy =1$, where we weight unit $i$ by $\hpi^{-1}$.
Then $\hat \gamma_1$ and $\hat \gamma_0$ coincide with 
the coefficient vectors of $x_i - \bar x$ from the weighted-least-squares fits of 
\beginy\label{eq:wls_x_1_2}
\lmt\{Y_i \sim 1 + (x_i - \bar x)\} \quad\text{over $\{i: \riy = 1, Z_i = 1\}$ with weight $\hpii$ for unit $i$;}\\
 \label{eq:wls_x_0_2}
\lmt\{Y_i \sim 1 + (x_i - \bar x)\}\quad\text{over $\{i: \riy =1,  Z_i = 0\}$ with weight $\hpii$ for unit $i$,}
\endy
respectively.   
Let $\hmuox$ and $\hmuzx$ denote the intercepts from \eqref{eq:wls_x_1_2} and  \eqref{eq:wls_x_0_2}, respectively. 

\begin{lemma}\label{lem:xreg_num}
$\htreg  = \hmuox - \hmuzx$ with 
\begina
 \hmuox &=&  \dfrac{  \sumi \rp \cdot Z_i \cdot\left\{Y_i(1)- (x_i - \bar x)^\T \hgo  \right\} }{\sumi \rp\cdot Z_i },\\
\hmuzx &=& \dfrac{\sumi  \rp \cdot(1-Z_i) \cdot \left\{Y_i(0) - (x_i - \bar x)^\T \hgz  \right\} }{\sumi  \rp\cdot (1-Z_i)}.
\enda

\end{lemma}

Recall the definition of $\hmuoipw $ and $\hmuzipw$ from Proposition \ref{prop:hajek_app}. 
From Lemma \ref{lem:xreg_num}, $\hmuox$ and $\hmuzx$  are analogs of $\hmuoipw $ and $\hmuzipw  $ defined on the adjusted potential outcomes $Y_i(1) - (x_i - \bar x)^\T \hgo $ and $Y_i(0) - (x_i - \bar x)^\T \hgz $, respectively.

\begin{proof}[Proof of Lemma \ref{lem:xreg_num}] 
That $\hat\tau_\xreg = \hmuox - \hmuzx$ follows from properties of least squares. 
To verify the explicit form of $\hmuox$ and $\hmuzx$, observe that the residual from \eqref{eq:wls_x_1_2} equals 
$
Y_i(1) - (x_i - \bar x)^\T \hgo  - \hywlsxo $ for units with $\riy = 1$ and $Z_i = 1$.
The first-order condition ensures
\begina
\sumi \rp \cdot  Z_i \cdot\left\{Y_i(1) - (x_i - \bar x)^\T \hgo  - \hywlsxo \right\} = 0
\enda
now that \eqref{eq:wls_x_1_2} includes an intercept. This verifies the expression of $\hmuox$. The expression of $\hmuzx$  follows by symmetry.

\end{proof}

\begin{lemma}\label{lem:mle} 
Consider a population of $N$ units, indexed by $i = \ot{N}$, each with a binary response $Z_i \in \{1,0\}$. 
Let $(x_i)_{i=1}^N$  and $(x_i')_{i=1}^N$, where $x_i, x_i'\in \mathbb R^J$, denote two sets of covariate vectors that satisfy $x_i' = \Gamma x_i$ for some nonsingular matrix $\Gamma \in \mathbb R^{J\times J}$. 
Let $\hat\beta$ and $\hat\beta'$ denote the maximum likelihood estimates of the coefficient vectors of $x_i$ and $x_i'$ from the logistic regressions 
$\glmt(Z_i \sim 1 + x_i)$ and $\glmt(Z_i \sim 1+x_i')$, 
respectively, and let $\hat p_i = \pr(Z_i = 1\mid x_i, \hat\beta) $ and $\hat p_i' =\pr(Z_i = 1\mid x_i, \hat\beta')  $ denote the corresponding estimated probabilities of $Z_i$ being one. 
Then   
 $\hat\beta = \Gamma^\T\hat\beta'$ and $\hat p_i = \hat p_i'$. 
\end{lemma}

The proof of Lemma \ref{lem:mle} follows from standard properties of logistic regression. We omit the details.

\section{Proofs of Theorems \ref{thm:clt} and \ref{thm:clt_app}}
\label{sec:proof_thm1}
We verify in this section Theorem \ref{thm:clt_app} as the complete version of Theorem \ref{thm:clt}. To add clarity, we use the superscript $*$ to indicate true values of the parameters throughout this section.
In particular, let
\begina
&&\muos  = E\big\{\yio\big\}, \qquad \muzs  = E\big\{\yiz\big\}, \qquad \muxs = E(x_i), \\
&& \gos  = \covxinv \cov\{x_i, \yio\}, \qquad \gzs  = \covxinv \cov\{ x_i, \yiz\}. 
\enda
Recall $\beta^*$ as the true value of the parameter in \eqref{eq:p_logit} under Assumption \ref{assm:riy} with $p(x_i, Z_i; \beta^*) = p_i$.
Recall $\alpha^*$ as the true value  of the parameter in \eqref{eq:e_logit} under Bernoulli randomization with $ e(x_i; \alpha^*) = e$. 
Recall $\hat\beta$ and $\hat\alpha$ as the maximum likelihood estimates of $\beta^*$ and $\alpha^*$, respectively, with $\hpi = p(x_i, Z_i; \hat\beta)$ and $\hei = e(x_i; \hat\alpha)$. 

Sections \ref{sec:unadj}--\ref{sec:reg} give the proofs of the results for $\htn$, $\hteps$, and $\htreg$, respectively. 
Section \ref{sec:thm1_special cases} gives the proof of the three special cases. 

\subsection{Proof of the result for $\htn $}\label{sec:unadj}
\begin{proof}
Recall from Proposition \ref{prop:hajek_app} that $\htn  = \hmuo  - \hmuz  $ with \begina
\hyua(1) = \dfrac{\sumi Z_i \cdot \dfrac{\riy}{\hpi} \cdot Y_i }{\sumi Z_i \cdot \dfrac{\riy}{\hpi}}, \quad \hyua(0) = \dfrac{\sumi  (1-Z_i)  \cdot\dfrac{\riy}{\hpi} \cdot  Y_i}{\sumi (1-Z_i)  \cdot\dfrac{\riy}{\hpi}}.
\enda 
Then $\mu_1 = \hmuo $ and $\mu_0 = \hmuz  $ solve  
\begina
0 =  \sumi Z_i \cdot \rp \cdot \left\{Y_i(1) - \mu_1\right\} \quad \text{and}\quad  
0 =  \sumi (1-Z_i) \cdot \rp \cdot \left\{Y_i(0) - \mu_0 \right\},
\enda
respectively, where $\hat p_i = p(x_i, Z_i; \hat\beta)$.
This, together with Lemma \ref{lem:beta}, ensures that $(\mu_1, \mu_0, \beta) = (\hmuo , \hmuz  , \hat\beta)$ jointly solves 
\begina
0 = \meani \psi \left(\yio, \yiz, x_i, Z_i, \ryi; \mu_1, \mu_0, \beta \right) = \meani \beginp  \psi_1\left(\yio, x_i, Z_i, \ryi; \mu_1,  \beta \right) \\ \psi_0 \left(\yiz, x_i, Z_i, \ryi; \mu_0,   \beta\right)  \\ \psi_\beta(x_i, Z_i, \ryi ; \beta) \endp,
\enda
where \renewcommand{\arraystretch}{1.5}
\beginy\label{eq:ee_wls}
\psi = \beginp  \psi_1 \\ \psi_0 \\ \psi_\beta\endp
\quad \text{with} \quad 
\begin{array}{lll}
 \psi_1& = &
 Z_i \cdot \rpfun \cdot\big\{Y_i(1) - \mu_1 \big\},\\
  \psi_0 &=&
 (1-Z_i) \cdot \rpfun  \cdot\big\{Y_i(0) - \mu_0 \big\},\\
\psi_\beta  &=& U_i \big\{\riy - p(x_i, Z_i; \beta) \big\}.
\end{array}
\endy
Direct algebra ensures that $(\mu_1, \mu_0, \beta) = (\muos, \muzs, \beta^*)$ solves $E\{\psi(\mu_1, \mu_0, \beta)\} = 0$.
The theory of m-estimation ensures  
\beginy\label{eq:clt_ipw}
\sqrt N\left\{\beginp 
\hmuo \\
\hmuz  \\
\hat\beta
\endp - \beginp \muos \\\muzs \\\beta^* \endp  \right\}  \rs  \mn\left\{ 0, \aba \right\},
\endy
where $A^*_\ua  $ and $B_\ua^* $ are the values of $ - E\left(\pd{(\mu_1, \mu_0, \beta)} \psi \right)$ and $ E(\psi \psi^\T)$ evaluated at $(\mu_1, \mu_0, \beta) =(\muos , \ \muzs , \ \beta^*)$, respectively. We compute below $A^*_\ua  $ and $B_\ua^* $, respectively.

\bigskip
\noindent\underline{\textbf{Compute  $B_\ua^*$.}}

\smallskip

Let $(\psi^*, \psi_1^*, \psi_0^*, \psi_\beta^*)$ denote the value of $(\psi, \psi_1, \psi_0, \psi_\beta)$ evaluated at $(\mu_1, \mu_0, \beta) = (\muos , \muzs , \beta^*)$. 
From \eqref{eq:ee_wls}, we have 
\beginy\label{eq:psistar_wls}
\psi^* = \beginp  \psi^*_1 \\ \psi^*_0 \\ \psi^*_\beta\endp
\quad \text{with} \quad 
\begin{array}{lll}
 \psi_1^*& = &
 Z_i \cdot \rpn \cdot \dyo ,\\
  \psi_0^* &=&
 (1-Z_i) \cdot \rpn \cdot \dyz,\\
\psi_\beta^*  &=& U_i (\ryi - p_i).
\end{array}
\endy
The $B_\ua^*$ matrix equals  
\beginy\label{eq:B_wls}
B_\ua^* = \left.E\left(\psi \psi^\T\right)\right|_{(\muos ,\,  \muzs ,\,  \beta^*)} = E\left(\psi^* \psi^{*\T}\right) = \left(\begin{array}{cc|c}
b_{11} & 0 & b_{13} \\
0 & b_{22} & b_{23} \\\hline
b_{13}^\T & b_{23}^\T & b_{33} 
\end{array}\right)
=
\left(\begin{array}{c|c}
B_{11} & B_{12} \\\hline
B_{12}^\T & B_{22}
\end{array}\right) ,
\endy
where 
\begina
B_{11} = \beginp b_{11}& 0 \\ 0 & b_{22} \endp, \quad B_{12} = \beginp b_{13}\\b_{23}\endp, \quad B_{22} = b_{33}
\enda
with 
\beginy\label{eq:bs}
b_{11} =  E(\psi_1^* \psi_1^{*\T}) &=&   E\left[   Z_i \cdot \dfrac{\riy}{ p_i^2} \cdot\big\{Y_i(1)- \muos \big\}^2 \right] = 
 E\left(   Z_i \cdot  \dfrac{E\left(\riy \mid x_i, Z_i\right)}{ p_i^2}   \cdot E\left[\left.\big\{Y_i (1) - \muos  \big\}^2 \ \right| \ x_i, Z_i\right]  \right) \nonumber\\
 &=& 
 E\left[ Z_i \cdot p_i^{-1}  \cdot  \dyo ^2   \right],\nonumber\\
b_{13} =  E(\psi_1^* \psi_\beta^{*\T})
&=&E \left[ Z_i \cdot \rpn \cdot \big\{Y_i(1)- \muos \big\}  \cdot  (\riy - p_i)U_i^\T \right] 
=  E  \left[Z_i \cdot \dfrac{\riy}{ p_i} \cdot (1-p_i) \cdot \big\{Y_i(1)- \muos \big\}  \cdot U_i^\T  \right] \nonumber\\
  &=& E  \left[Z_i \cdot (1-p_i) \cdot  \big\{Y_i(1)- \muos \big\}  \cdot U_i^\T  \right],\nonumber\\
   b_{22}=  E(\psi_0^* \psi_0^{*\T}) &=& 
 E\left[ (1-Z_i) \cdot p_i^{-1}  \cdot  \dyz  ^2   \right]\quad \text{by symmetry},\nonumber\\
b_{23} =  E(\psi_0^* \psi_\beta^{*\T})
   &=& E  \left[  (1-Z_i) \cdot(1-p_i) \cdot \big\{Y_i(0)- \muzs \big\}  \cdot U_i^\T  \right] \quad \text{by symmetry},\nonumber\\
   b_{33} =  E(\psi_\beta^* \psi_\beta^{*\T}) &=&   E\left\{p_i(1-p_i) \cdot U_i U_i^\T \right\}\quad \text{by Lemma \ref{lem:beta}}.
\endy

\bigskip
\noindent\underline{\textbf{Compute  $A_\ua^*$.}} 

\smallskip

With a slight abuse of notation, let 
$
\pd{\mu_1} \psi_1^* = \left. \pd{\mu_1}\psi_1\right|_{(\muos ,\, \muzs ,\, \beta^*)}
$
denote the value of $\pd{\mu_1} \psi_1 $ evaluated at $(\mu_1, \mu_0, \beta)=(\muos ,\, \muzs ,\, \beta^*)$. Similarly define other partial derivatives. 
From \eqref{eq:ee_wls}, we have 
\begina
\left.\pd{(\mu_1, \mu_0, \beta)} \psi \right|_{(\muos , \,  \muzs ,\, \beta^*)}
&=& 
\beginp
\pd{\mu_1}\psi_1^*  & 0  &  \pd{\beta^\T} \psi_1^* \\
0 & \pd{\mu_0}\psi_0^*  &   \pd{\beta^\T} \psi_0^* \\
0 & 0 & \pd{\beta^\T} \psi_\beta^*
\endp,
\enda
where
\begina
\pd{\mu_1}\psi_1^* &=& - Z_i \cdot \rpn ,\\
\pd{\beta^\T} \psi_1^*
 &=& Z_i \cdot \riy \cdot  \dyo   \cdot \pd{\beta^\T} p_i^{-1}\\
& =& -   Z_i \cdot\riy  \cdot  \dyo  \cdot \dfrac{1-p_i}{ p_i}  U_i^\T 
=  - Z_i \cdot \rpn  (1-p_i)\cdot \dyo  \cdot U_i^\T.
\enda
Accordingly,  we have 
\begina
A_\ua^* 
= -\left.E\left( \pd{(\mu_1, \mu_0, \beta)} \psi\right) \right|_{(\muos , \,  \muzs ,\, \beta^*)} 
= -E\left( \left.\pd{(\mu_1, \mu_0, \beta)} \psi \right|_{(\muos , \,  \muzs ,\, \beta^*)}\right)= \beginp
a_{11} & 0 & a_{13} \\
0 & a_{22} & a_{23} \\
0 & 0 & a_{33}
\endp 
\enda
where
\beginy\label{eq:as}
a_{11} &=& - E \left(\pd{\mu_1}\psi_1^*\right) = E\left(  Z_i\cdot \rpn
\right) = E(Z_i) = e, \nonumber\\
a_{22} &=& - E \left(\pd{\mu_0}\psi_0^*\right) =  1-e\bsm,\nonumber \\\nonumber\\
a_{13} &=& - E\left(\pd{\beta^\T} \psi_1^* \right)\nonumber\\
&=& E \left[ Z_i \cdot\rpn(1-p_i)\cdot  \dyo  \cdot  U_i^\T \right]  = E \left[ Z_i \cdot(1-p_i)\cdot  \dyo  \cdot  U_i^\T \right] = b_{13},\nonumber\\
a_{23} &=& - E\left(\pd{\beta^\T} \psi_0^* \right)  =E \left[  (1-Z_i) \cdot (1-p_i)  \cdot\dyz   \cdot  U_i^\T\right] = b_{23}\bsm,\nonumber\\\nonumber\\
a_{33} &=&  - E \left(    \pd{\beta^\T} \psi_\beta^*\right) = E\left\{ p_i(1-p_i) \cdot U_i U_i^\T \right\} = b_{33}  \blb.
\endy
This ensures 
\beginy\label{eq:A_wls}
A_\ua^* =\left(\begin{array}{cc|c}
a_{11} & 0 & b_{13} \\
0 & a_{22} & b_{23} \\\hline 
0 & 0 & b_{33}
\end{array}\right) = \left(\begin{array}{c|c}
A_{11} & B_{12} \\\hline 
0 & B_{22}
\end{array}\right) \quad \text{with} \quad
A_\ua^{*-1} 
=  \beginp
A_{11}^{-1} & -A_{11}^{-1}B_{12} B_{22}^{-1} \\
0 &B_{22}^{-1}
\endp,
\endy
where $A_{11} = \beginp a_{11} & 0 \\ 0 & a_{22} \endp = \beginp e & 0 \\ 0 & 1-e \endp$ and $(B_{12}, B_{22})$ are defined in \eqref{eq:B_wls}.

\bigskip 
\bigskip

\noindent\underline{\textbf{Compute $v_\ua$.}}

\smallskip

Equations \eqref{eq:clt_ipw},  \eqref{eq:B_wls}, and \eqref{eq:A_wls} together ensure
\begina
V_\ua  &=& \cova\left\{
\beginp
\hmuo \\
\hmuz  
\endp 
\right\} = \cova\left\{\beginp 1&0& 0 \\  0&1&0 \endp \beginp\hmuo\\ \hmuz \\ \hat\beta\endp \right\}\\ 
&\overset{\eqref{eq:clt_ipw}}{=}& \beginp I_2& 0  \endp\aba\beginp I_2 \\ 0  \endp \\
&\overset{\eqref{eq:B_wls},  \eqref{eq:A_wls}}{=}& \beginp
A_{11}^{-1} & -A_{11}^{-1}B_{12} B_{22}^{-1}
\endp
\beginp
B_{11} & B_{12} \\
B_{12}^\T & B_{22}
\endp 
\beginp
A_{11}^{-1} \\
- B_{22}^{-1}B_{12}^\T A_{11}^{-1}
\endp\\
&=&
A_{11}^{-1}B_{11}A_{11}^{-1}  - A_{11}^{-1}B_{12}B_{22}^{-1} B_{12}^\T A_{11}^{-1} \\
&=&
\beginp
\dfrac{b_{11}}{e^2} & 0  \\ 
0 & \dfrac{b_{22}}{(1-e)^2} 
\endp
- 
\beginp
\dfrac{b_{13}}{e}\\
\dfrac{b_{23}}{1-e}
\endp
b_{33}^{-1} \left(\dfrac{b_{13}^\T}{e}, \dfrac{b_{23}^\T}{1-e}\right).   
\enda 
This ensures
\begina
v_\ua=\va(\hat\tau_\ua ) &=& \va\left\{\hmuo  - \hmuz \right\}\\ 
&=& (1,-1)  V_\ua  \beginp
1\\
-1
\endp \\
&=& (1,-1)  \left\{\beginp
\dfrac{b_{11}}{e^2} & 0  \\ 
0 & \dfrac{b_{22}}{(1-e)^2} 
\endp
- 
\beginp
\dfrac{b_{13}}{e}\\
\dfrac{b_{23}}{1-e}
\endp
b_{33}^{-1} \left(\dfrac{b_{13}^\T}{e}, \dfrac{b_{23}^\T}{1-e}\right)\right\} \beginp
1\\
-1
\endp \\
&=&  (1,-1)  \beginp
\dfrac{b_{11}}{e^2} & 0  \\ 
0 & \dfrac{b_{22}}{(1-e)^2} 
\endp \beginp
1\\
-1
\endp -  (1,-1)  \beginp
\dfrac{b_{13}}{e}\\
\dfrac{b_{23}}{1-e}
\endp
b_{33}^{-1} \left(\dfrac{b_{13}^\T}{e}, \dfrac{b_{23}^\T}{1-e}\right)
 \beginp
1\\
-1
\endp\\
&=&  \dfrac{b_{11}}{e^2} + \dfrac{b_{22}}{(1-e)^2} - \left(\dfrac{b_{13}}{e}-\dfrac{b_{23}}{1-e}\right)b_{33}^{-1}\left(\dfrac{b_{13}}{e}-\dfrac{b_{23}}{1-e}\right)^\T.
\enda

\end{proof}

\subsection{Proof of the result for $\hteps $}\label{sec:ipwx_proof}
\begin{proof}
Recall from Proposition \ref{prop:hajek_app} that $\hteps = \hmuoipwx  - \hmuzipwx $, where
\begina
\hmuoipwx  = \dfrac{\sumi \rp \cdot \zhe \cdot Y_i(1)}{\sumi \rp \cdot \zhe} ,\qquad
\hmuzipwx  = \dfrac{\sumi \rp \cdot \omzhe \cdot Y_i(0)}{\sumi \rp \cdot \omzhe}.
\enda 
Then $\mu_1 = \hmuoipwx $ and $\mu_0 = \hmuzipwx $ solve 
\begina
0  = \sumi \rp \cdot \zhe \cdot \left\{Y_i(1) - \mu_1\right\}\quad \text{and}\quad 
0 =  \sumi \rp \cdot \omzhe \cdot \left\{Y_i(0) - \mu_0 \right\},
\enda
respectively, 
where $\hat p_i = p(x_i, Z_i; \hat\beta)$ and $\hei  = e(x_i, \hat\alpha)$. 
This, together with Lemma \ref{lem:beta}, ensures that $(\mu_1, \mu_0, \beta, \alpha) =(\hyipwxo , \hyipwxz  , \hat\beta, \hat\alpha)$ jointly solves 
\begina
0 = \meani \phi \big(Y_i(1), Y_i(0), x_i, Z_i, \ryi; \mu_1, \mu_0, \beta, \alpha \big),
\enda
where \renewcommand{\arraystretch}{1.5}
\beginy\label{eq:ee_ipwx}
\phi = \beginp  \phi_1 \\ \phi_0 \\ \psi_\beta\\\psi_\alpha\endp 
\quad\text{with} \quad
\begin{array}{ccl}
\phi_1 &=& 
 \rpfun  \cdot \zef \cdot \big\{Y_i(1) - \mu_1 \big\}, \\
\phi_0 &=& \rpfun  \cdot \omzef \cdot \big\{Y_i(0) - \mu_0 \big\},\\
\psi_\beta &=&  U_i \big\{\riy - \pfun \big\},\\
\psi_\alpha&=& \txi \big\{Z_i - \ef \big\}.
\end{array}
\endy
We use the letter $\phi$ to highlight differences from \eqref{eq:ee_wls}.
Direct algebra ensures that $(\mu_1, \mu_0, \beta, \alpha) = (\muos, \muzs, \beta^*, \alpha^*)$ solves $E\{\phi (\mu_1,\mu_0, \beta,\alpha)\} = 0$.
The theory of m-estimation ensures 
\beginy\label{eq:clt_ipwx}
\sqrt N\left\{\beginp 
\hyipwxo \\
\hyipwxz  \\
\hat\beta\\
\hat\alpha
\endp - \beginp \muos \\\muzs \\\beta^* \\ \alpha^*\endp  \right\} \rs \mn\left\{ 0, \abaeps \right\},
\endy
where $A^*_\xeps$ and $B^*_\xeps$ are the values of $ - E\left(\pd{(\mu_1, \mu_0, \beta, \alpha)} \phi \right)$ and $ E(\phi \phi^\T)$ evaluated at $(\muo, \muz,  \beta, \alpha)=(\muos ,  \muzs , \beta^*, \alpha^*)$, respectively. We compute below $A^*_\xeps$ and $B^*_\xeps$, respectively.

\bigskip 
\noindent\underline{\textbf{Compute $B_\xeps^*$.}}

\smallskip

Recall that $p(x_i, Z_i; \beta^*) = p_i$  and $ e(x_i; \alpha^*)=e$. 
Let $(\phi^*, \phi_1^*, \phi_0^*, \psi_\beta^*, \psi_\alpha^*)$ denote the value of $(\phi, \phi_1, \phi_0, \psi_\beta, \psi_\alpha)$ evaluated at $(\muos , \, \muzs , \, \beta^*, \alpha^*)$. 
From \eqref{eq:psistar_wls} and \eqref{eq:ee_ipwx}, we have 
\begina   
\phi^* = \beginp  \phi^*_1 \\ \phi^*_0 \\ \psi^*_\beta\\\psi^*_\alpha\endp 
\quad\text{with} \quad
\begin{array}{lll}
\phi_1^* &=& 
 \rpn  \cdot \ze \cdot \dyo = e^{-1}\psi_1^*, \\
\phi_0^* &=& \rpn  \cdot \omze \cdot \dyz = (e-1)^{-1}\psi_0^*,\\
\psi_\beta^* &=&  U_i \big(\riy - p_i \big),\\
\psi_\alpha^*&=& \txi \big(Z_i - e \big).
\end{array}
\enda
The $B_\xeps^*$ matrix equals  
\beginy\label{eq:B_ipwx}
B^*_\xeps = \left. E\left(\phi \phi^\T\right)\right|_{(\muos ,\,  \muzs , \,  \beta^*, \, \alpha^*)} = E \left( \phi^* \phi^{*\T} \right) = \left(\begin{array}{cc|cc}
c_{11} & 0 & c_{13} & c_{14} \\
0 & c_{22} & c_{23} & c_{24} \\\hline
c_{13}^\T & c_{23}^\T & c_{33} &c_{34}\\
c_{14}^\T & c_{24}^\T & c_{34}^\T & c_{44}
\end{array}\right) =   \left(\begin{array}{c|c}
C_{11} & C_{12} \\\hline
C_{12}^\T & C_{22}
\end{array}\right) ,
\endy
where 
\begina
C_{11} = \beginp c_{11} & 0 \\ 0 & c_{22} \endp, \quad C_{12} = 
 \beginp
c_{13} & c_{14} \\
c_{23} & c_{24} \\ 
\endp, \quad C_{22} = \beginp c_{33} & c_{34}\\ c_{34}^\T & c_{44} \endp = \beginp c_{33} &0 \\ 0 & c_{44} \endp
\enda
with \renewcommand{\arraystretch}{1.5}
\beginy\label{eq:cs}
c_{11} &=& E(\phi_1^* \phi_1^*) = e^{-2} E(\psi_1^* \psi_1^*) = e^{-2} b_{11}, \nonumber\\
 c_{22} &=& E(\phi_0^* \phi_0^*) = (1-e)^{-2} E(\psi_0^* \psi_0^*)= (1-e)^{-2} b_{22},\nonumber\\
c_{13} &=& E(\phi_1^* \psi_\beta^{*\T}) = e^{-1}E(\psi_1^* \psi_\beta^{*\T}) = 
   e^{-1} b_{13},\nonumber\\
 c_{23} 
   &=& E(\phi_0^* \psi_\beta^{*\T}) = (1-e)^{-1}E(\psi_0^* \psi_\beta^{*\T}) =  (1-e)^{-1} b_{23},\nonumber\\
   c_{33} &=& E(\psi_\beta^* \psi_\beta^{*\T}) =b_{33}
   \endy
   by the definitions of $(b_{11}, b_{22}, b_{13}, b_{23}, b_{33})$ from \eqref{eq:bs} and 
   \begina
c_{14} 
&=& E(\phi_1^* \psi_\alpha^{*\T}) = E \left[ \rpn  \cdot \ze\cdot \big\{Y_i(1)- \muos \big\}  \cdot  (Z_i - e)\txi^\T \right] 
=  E \left[   \dfrac{Z_i(1-e)}{e} \cdot \big\{Y_i(1)- \muos \big\}  \cdot   \txi^\T \right]  \\
  &=& (1-e) \cdot E \left[    \big\{Y_i(1)- \muos \big\}  \cdot   \txi^\T \right] ,\\
    c_{24} 
&=& E(\phi_0^* \psi_\alpha^{*\T}) = E \left[ \rpn  \cdot \omze\cdot \big\{Y_i(0)- \muzs \big\}  \cdot  (Z_i - e)\txi^\T \right] 
=  E \left[   \dfrac{-e(1-Z_i)}{1-e} \cdot \big\{Y_i(0)- \muzs \big\}  \cdot   \txi^\T \right]  \\
  &=& -e \cdot E \left[    \big\{Y_i(0)- \muzs \big\}  \cdot   \txi^\T \right],\\
      c_{34} &=& E\left\{U_i   (R_i^Y - p_i) \cdot (Z_i - e) \txi^\T \right\} = 0,\\
    c_{44} &=&  E(\psi_\alpha^* \psi_\alpha^{*\T}) = E\left\{ \txi \txi^\T (Z_i - e)^2 \right\} = e(1-e) \cdot E\left(\txi \txi^\T  \right).  
\enda

\bigskip 
\noindent\underline{\textbf{Compute $A_\xeps^*$.}}

\smallskip 

With a slight abuse of notation,  let 
$
\pd{\mu_1} \phi^*_1 =  \left.\pd{ \mu_1 } \phi_1 \right|_{(\muos ,\,  \muzs , \,  \beta^*, \, \alpha^*)}
$ denote the partial derivative of $\phi_1$ with regard to $\mu_1$  evaluated at $(\muos , \muzs , \beta^*, \alpha^*)$. Similarly define $\pd{ \mu_0 } \phi_1^*$, $\pd{ \mu_1 } \psi_0^*$, etc.
Recall that 
\begina
\pd{\beta} p_i = \left.\pd{ \beta } \pfun \right|_{\beta = \beta^*}= p_i(1-p_i)U_i, \qquad 
\pd{\beta} p_i^{-1} = \left.\pd{ \beta } \pfun^{-1} \right|_{\beta = \beta^*}=- \dfrac{1-p_i}{p_i} U_i
\enda
from \eqref{eq:pdb}. 
By symmetry, we have 
\beginy\label{eq:pdalpha}
 \left.\pd{ \alpha} \ef   \right|_{\alpha= \alpha^* }\ = \ e(1-e)\txi, \qquad 
  \left.\pd{ \alpha} \ef^{-1} \right|_{\alpha= \alpha^* }=-\dfrac{1-e}{e} \txi, \qquad
   \left.\pd{ \alpha} \{1-\ef\}^{-1} \right|_{\alpha= \alpha^* }=\dfrac{e}{1-e} \txi. \qquad
\endy

From \eqref{eq:ee_ipwx}, we have 
\begina
 \left.\pd{(\mu_1, \mu_0, \beta, \alpha)} \phi \right|_{(\muos ,\,  \muzs , \,  \beta^*, \, \alpha^*)}
&=& 
\beginp
\pd{\mu_1} \phi_1^*  & 0  &   \pd{\beta^\T} \phi_1^*  &    \pd{\alpha^\T} \phi_1^*  \\
0 & \pd{\mu_0} \phi_0^*  &   \pd{\beta^\T} \phi_0^*   & \pd{\alpha^\T} \phi_0^*  \\
0 & 0 &   \pd{\beta^\T} \psi_\beta^* & 0\\
0 & 0 & 0 &   \pd{\alpha^\T} \psi_\alpha^* 
\endp,
\enda
where
 \begina
  \pd{\alpha^\T} \phi_1^*
 &=&   \rpn  \cdot Z_i \cdot \dyo   \left.\pd{\alpha^\T} \ef^{-1} \right|_{\alpha^*} =  - \rpn  \cdot \ze  \cdot (1-e) \cdot \dyo  \cdot \txi^\T,\\
 \pd{\alpha^\T} \phi_0^*
 &=&  \rpn  \cdot (1-Z_i) \cdot \dyz    \cdot  \left.\pd{ \alpha^\T} \{1-\ef\}^{-1} \right|_{\alpha= \alpha^* }= \rpn  \cdot \omze \cdot e \cdot \dyz     \cdot   \txi^\T,\\
 \pd{\alpha^\T} \psi_\alpha^* &=& -  \txi \left.\pd{\alpha^\T} \ef \right|_{\alpha^*} = - \txi \txi^\T e(1-e) 
\enda
by \eqref{eq:pdalpha}.
Accordingly,  we have 
\begina
A_\xeps^* = - \left.E\left(\pd{(\mu_1, \mu_0, \beta, \alpha)} \psi \right)\right|_{(\muos , \, \muzs , \, \beta^*, \, \alpha^*)} = - E\left(\left.\pd{(\mu_1, \mu_0, \beta, \alpha)} \phi \right|_{(\muos ,\,  \muzs , \,  \beta^*, \, \alpha^*)} \right) =  \beginp
a'_{11} & 0 & a'_{13} & a_{14} \\
0 & a'_{22} & a'_{23} & a_{24} \\
0 & 0 & a'_{33} & 0\\
0 & 0 & 0 & a_{44}
\endp 
\enda
where
\begina
a'_{11} &=& - E\left(\pd{\mu_1} \phi_1^*\right) = e^{-1}E\left(\pd{\mu_1} \psi_1^*\right) =  e^{-1} a_{11} = 1,\\
 a'_{13} &=& - E \left(\pd{\beta^\T} \phi_1^* \right) = e^{-1} E \left(\pd{\beta^\T} \psi_1^*\right) = e^{-1}a_{13} = c_{13},\\
a'_{22} &=& 1, \quad 
a'_{23}= c_{23} \quad \text{ by symmetry},\\
a_{33}'&=&  -E\left( \pd{\beta^\T} \psi_\beta^* \right) = c_{33}
\enda
by \eqref{eq:as} and \eqref{eq:cs}, and 
\begina
a_{14} &=&  - E\left( \pd{\alpha^\T} \phi_1^* \right) 
= E\left[ \rpn  \cdot \ze  \cdot (1-e) \cdot \dyo  \cdot \txi^\T \right]\\
&=& E\left[   \ze  \cdot (1-e) \cdot \dyo  \cdot \txi^\T \right] = (1-e) \cdot  E\left[   \dyo  \cdot \txi^\T \right] = c_{14},\\
a_{24} &=&  - E\left( \pd{\alpha^\T} \phi_0 ^*  \right) 
= - E\left[ \rpn  \cdot \omze \cdot e \cdot \dyz   \cdot   \txi^\T \right]\\
&=& - E\left[   \omze \cdot e \cdot \dyz   \cdot   \txi^\T \right] = - e \cdot  E\left[   \dyz   \cdot \txi^\T \right] = c_{24},\\
a_{44}& =& - E \left(\pd{\alpha^\T} \psi_\alpha^*\right)   = c_{44}. 
\enda
This, together with the definitions of $C_{12}$ and $ C_{22}$ after \eqref{eq:B_ipwx}, ensures 
\beginy\label{eq:A_ipwx}
A^*_\xeps   = \left(\begin{array}{cc|cc}
1 & 0 & c_{13} & c_{14} \\
0 & 1 & c_{23} & c_{24} \\ \hline
0 & 0 & c_{33} & 0\\
0 & 0 & 0 & c_{44}
\end{array}\right) = 
\left(\begin{array}{c|c}
I_2 & C_{12}\\ \hline 
0 & C_{22}
\end{array}\right) \quad \text{with} 
\quad 
A_\xeps^{*-1} 
=  \beginp
I_2 & - C_{12} C_{22}^{-1} \\
0 & C_{22}^{-1} 
\endp. 
\endy

\bigskip 
\noindent\underline{\textbf{Compute $v_\xeps$.}}

\smallskip

Equations \eqref{eq:clt_ipwx},  \eqref{eq:B_ipwx}, and \eqref{eq:A_ipwx} together ensure
\beginy\label{eq:Vipwx}
 V_\ipwx &=&  \cova\left\{
\beginp
\hmuoipwx  \\
\hmuzipwx  
\endp 
\right\} 
= \cova\left\{\beginp 1&0& 0 & 0 \\  0&1& 0 &0\endp \beginp\hmuoipwx\\ \hmuzipwx \\ \hat\beta \\ \hat\alpha\endp \right\}\nonumber\\
&\overset{\eqref{eq:clt_ipwx}}{=}& 
\beginp I_2  \ 0  \endp \abaeps 
\beginp I_2 \\  0 \endp \nonumber\\
&\overset{\eqref{eq:B_ipwx}, \eqref{eq:A_ipwx}}{=}&\beginp
I_2  \  - C_{12} C_{22}^{-1}
\endp
\beginp
C_{11} &  C_{12}  \\
C_{12}^\T & C_{22} 
\endp
\beginp
I_2\\
- C_{22}^{-1}C_{12}^{\T} 
\endp =
C_{11} - C_{12}C_{22}^{-1}C_{12}^{\T}. 
\endy  
The definitions of $C_{12}$ and $ C_{22}$ after \eqref{eq:B_ipwx} ensure 
\beginy\label{eq:bbc_ipwx}
C_{12}C_{22}^{-1}C_{12}^{\T} 
&=&
 \beginp
c_{13} & c_{14} \\
c_{23} & c_{24} \\ 
\endp  \beginp
c_{33}^{-1} & 0 \\
0 & c_{44}^{-1} \\ 
\endp \beginp
c_{13}^\T &c_{23}^\T \\
 c_{14}^\T  & c_{24}^\T \\ 
\endp \nonumber\\
&=&
 \beginp
c_{13}   \\
c_{23}  
\endp   
c^{-1}_{33}   \beginp
c_{13}^\T & c_{23}^\T
\endp 
+ 
 \beginp
c_{14} \\
c_{24} \\ 
\endp  c^{-1}_{44}   \beginp
c_{14}^\T & c_{24}^\T 
\endp.
\endy
This ensures
\begina
 v_\xeps  = \va(\hat\tau_\ipwx) 
&=&  \va\left\{\hmuoipwx  - \hmuzipwx \right\}\\ 
&=& (1,-1)  V_\ipwx \beginp
1\\
-1
\endp\\
&\overset{\eqref{eq:Vipwx}}{=}& (1,-1)  \left( C_{11}  -  C_{12}C_{22}^{-1}C_{12}^{\T} \right)  \beginp
1\\
-1
\endp\\
&=& (1,-1)  C_{11}  \beginp
1\\
-1
\endp - (1,-1)  C_{12}C_{22}^{-1}C_{12}^{\T}  \beginp
1\\
-1
\endp\\
\\ &
\overset{\eqref{eq:bbc_ipwx}}{=}& (1,-1) \beginp c_{11} & 0 \\ 0 & c_{22}\endp  \beginp
1\\
-1
\endp\\
&&   -(1,-1)  \beginp
c_{13}   \\
c_{23}  
\endp   
c^{-1}_{33}   \beginp
c_{13}^\T & c_{23}^\T
\endp \beginp
1\\
-1
\endp  
-(1,-1) \beginp
c_{14} \\
c_{24} \\ 
\endp  c^{-1}_{44}   \beginp
c_{14}^\T & c_{24}^\T 
\endp \beginp
1\\
-1
\endp \\
& =& c_{11}+ c_{22} - \left(c_{13}-c_{23}\right)c_{33}^{-1}\left(c_{13}-c_{23}\right)^\T -   \left(c_{14}-c_{24}\right)c_{44}^{-1}\left(c_{14}-c_{24}\right)^\T\\
&\overset{\eqref{eq:cs}}{=}& \frac{b_{11}}{e^2}+ \frac{b_{22}}{(1-e)^2} - \left(\frac{b_{13}}{e}-\frac{b_{23}}{1-e}\right)b_{33}^{-1}\left(\frac{b_{13}}{e}-\frac{b_{23}}{1-e}\right)^\T -   \left(c_{14}-c_{24}\right)c_{44}^{-1}\left(c_{14}-c_{24}\right)^\T\\
&=& v_\ua  -  \left(c_{14}-c_{24}\right)c_{44}^{-1}\left(c_{14}-c_{24}\right)^\T.
\enda

\bigskip 

\noindent\underline{\textbf{Simplify $\left(c_{14}-c_{24}\right)c_{44}^{-1}\left(c_{14}-c_{24}\right)^\T$.}}

\smallskip

By definition, 
$
 \{Y_i(1)- \muos \big\}  \cdot   \txi^\T = \{Y_i(1)- \muos \big\}  \cdot  (1, \ x_i^\T )  
 = ( Y_i(1)- \muos  , \ \{Y_i(1)- \muos \big\}  x_i^\T  ) 
$
with 
\begina
E \Big[ \big\{Y_i(1)- \muos \big\}  \cdot   \txi^\T  \Big] =\Big( E\dyo  , \ E \left[ \{Y_i(1)- \muos \big\}  x_i^\T \right] \Big)  = \Big(0, \ \cov\{\yio, x_i\} \Big).
 \enda
This ensures 
\begina
\begin{array}{llllllll}
c_{14} 
&=& (1-e) \cdot E \left[    \dyo  \cdot   \txi^\T \right]  &=& e(1-e) \cdot\Big( 0, \  \cov\big\{  e^{-1}Y_i(1) , x_i\big\} \Big ) ,\\
    c_{24} 
  &=& -e \cdot E \left[    \big\{Y_i(0)- \muzs \big\}  \cdot   \txi^\T \right] &=&-e(1-e) \cdot 
  \Big( 0, \  \cov\big\{ (1-e)^{-1} Y_i(0), x_i\big\}  \Big) \bsm
  \end{array}
  \enda
  with 
\beginy\label{eq:db}
\{e(1-e)\}^{-1}(c_{14} - c_{24}) = \Big(0,\  \cov\big(\tyi, x_i\big) \Big), \qquad \text{where} \ \ \tyi = \frac{Y_i(1)}{e} + \frac{\yiz}{1-e}. 
\endy
In addition, we have 
\begina
  c_{44} =e(1-e) \cdot E\left(\txi \txi^\T  \right)   =   e(1-e)  \beginp 1 & \muxst \\ \muxs & E(x_i x_i^\T)\endp 
\enda
with
\beginy\label{eq:bffinv}
e(1-e)c_{44}^{-1} 
&=& \beginp 1 & \muxst \\ \muxs & E(x_i x_i^\T)\endp ^{-1}\nonumber\\
&=& \beginp
1 + \muxst \left\{E(x_ix_i^\T) - \muxs \muxst \right\}^{-1} \muxs & - \muxst \left\{E(x_ix_i^\T) - \muxs \muxst \right\}^{-1}\nonumber\\
- \left\{E(x_ix_i^\T) - \muxs \muxst \right\}^{-1} \muxs & \left\{E(x_ix_i^\T) - \muxs \muxst \right\}^{-1}  
\endp \nonumber \\
&=& \beginp
1 + \muxst \covxinv  \muxs & - \muxst \covxinv \\
- \covxinv  \muxs & \covxinv 
\endp\nonumber\\
&=& \beginp
1 & 0\\
0 & 0 \endp + \beginp \muxst \\ -I_J \endp\covxinv  \left( \muxs, \  -I_J \right) 
\endy
by the formula for block matrix inverse. Equations \eqref{eq:db} and \eqref{eq:bffinv} together ensure
\begina
&&\{e(1-e)\}^{-1}\left(c_{14}-c_{24}\right)c_{44}^{-1}\left(c_{14}-c_{24}\right)^\T
\\
&=& \{e(1-e)\}^{-1}\left(c_{14}-c_{24}\right) \cdot \left\{e(1-e)c_{44}^{-1}\right\}\cdot \{e(1-e)\}^{-1}\left(c_{14}-c_{24}\right)^\T\\ 
&=& 
\Big(0,\  \cov\big(\tyi, x_i\big) \Big) \left[\beginp
1 & 0\\
0 & 0 \endp 
+
\beginp \muxst \\ -I_J \endp\covxinv  \left( \muxs, \  -I_J \right) \right]
\beginp
0\\
\cov\big(x_i, \tyi\big)
\endp\\
&=& 
\Big(0,\  \cov\big(\tyi, x_i\big) \Big) \beginp
1 & 0\\
0 & 0 \endp 
\beginp
0\\
\cov\big(x_i, \tyi\big)
\endp \\
&&+\Big(0,\  \cov\big(\tyi, x_i\big) \Big)
\beginp \muxst \\ -I_J \endp\covxinv  \left( \muxs, \  -I_J \right) 
\beginp
0\\
\cov\big(x_i, \tyi\big)
\endp\\
&=& \cov\big(\tyi, x_i\big) \covxinv\cov\big(x_i, \tyi\big)\\
&=& \var\{\proj(\tyi\mid 1,x_i)\}.
\enda
\end{proof}

\subsection{Proof of the result for $\htreg $}\label{sec:reg}
\begin{proof}
Recall from Lemma \ref{lem:xreg_num} that $(\hy_\xreg(1), \hat \gamma_1)$ and $(\hy_\xreg(0), \hat\gamma_0)$ are the intercepts and coefficient vectors of $(x_i-\bar x)$ from the weighted-least-squares fits of \eqref{eq:wls_x_1_2} and \eqref {eq:wls_x_0_2},
respectively, with $\htreg  = \hmuox - \hmuzx$.   
The first-order conditions of \eqref{eq:wls_x_1_2}--\eqref {eq:wls_x_0_2} ensure that $(\mu_1, \gamma_1) = (\hmuox, \hgo )$ and $(\mu_0, \gamma_0) = (\hmuzwlsx, \hgz ) $  solve   
\begina
0 &=&   \sumi 
Z_i \cdot \rp \cdot  \left\{Y_i(1)- (x_i - \bar x)^\T  \gamma_1 - \mu_1 \right\}\beginp 
1 \\
 x_i - \bar x 
\endp,\\
0  & =& \sumi (1-Z_i) \cdot \rp \cdot 
 \left\{Y_i(0)- (x_i - \bar x)^\T  \gamma_0 - \mu_0 \right\}\beginp
1\\
 x_i - \bar x 
\endp,
\enda
respectively, where $\hat p_i = p(x_i, Z_i; \hat\beta)$. 
This, together with Lemma \ref{lem:beta}, ensures that $( \mu_1,\gamma_1, \mu_0, \gamma_0,\mu_x, \beta) = (\hywlsxo  ,  \hgo , \hmuzwlsx  ,\hgz , \bar x, \hat\beta)$ jointly solves 
\begina
0 = \meani \eta \big(Y_i(1), Y_i(0), x_i, Z_i, \riy; \theta \big) = \meani  \beginp  \eta_1\left(Y_i(1), x_i, Z_i, \ryi; \mu_1, \gamma_1, \mux, \beta \right) \\  \eta_0\left(\yiz, x_i, Z_i, \ryi; \mu_0, \gamma_0, \mux,  \beta \right) \\  \psi_x(x_i ; \mux) \\ \psi_\beta(x_i, Z_i, \riy ; \beta) \endp,
\enda
where $\theta =( \mu_1,\gamma_1, \mu_0, \gamma_0,\mu_x, \beta)$ and 
\beginy\label{eq:ee_wlsx}
\eta = \beginp  \eta_1 \\  \eta_0 \\  \psi_x \\ \psi_\beta\endp
\quad \text{with} \quad 
\begin{array}{lll}
\eta_1 &=&Z_i \cdot \rpfun \cdot  \left\{Y_i(1)- (x_i - \mux)^\T  \gamma_1 - \mu_1 \right\} \beginp 
1 \\
 x_i -\mux 
\endp, \\ 
 \eta_0 &=&  (1-Z_i) \cdot \rpfun \cdot 
 \left\{Y_i(0)- (x_i - \mux)^\T  \gamma_0 - \mu_0 \right\}\beginp 
1\\
 x_i - \mux 
\endp,\\ 
\psi_x &=&  x_i - \mu_x,\\
 \psi_\beta &=& U_i \big\{\riy - p(x_i, Z_i; \beta) \big\}.
 \end{array}
\endy
Direct algebra ensures that $\theta = \theta^*= (\muos , \gos ,  \muzs , \gzs , \mu_x^*, \beta^*)$  solves $E\{\eta(\theta)\} = 0$. 
The theory of m-estimation ensures  
\beginy\label{eq:clt_wlsx}
\sqrt N\left\{\beginp 
\hywlsxo  \\
\hgo \\
\hmuzwlsx  \\
\hgz \\
\bar x\\
\hat\beta
\endp - 
\beginp \muos \\\gos \\\muzs \\\gzs \\ \mu_x^*\\\beta^* \endp  \right\}  \rs  \mn\left\{ 0, \abareg \right\},
\endy
where $A_\xreg^* $ and $B_\xreg^* $ are the values of $ - E\left(\pd{(\mu_1, \gamma_1,\mu_0,  \gamma_0, \mu_x, \beta)} \eta \right)$ and $ E(\eta \eta^\T)$ evaluated at $\theta = \theta^*$. We compute below $A_\xreg^*$ and $B_\xreg^*$, respectively.

\bigskip

\noindent\underline{\textbf{Compute   $B_\xreg^*$.}}

\smallskip

Let $(\eta^*, \eta_1^*, \eta_0^*, \psi_x^*, \psi_\beta^*)$ denote the value of $(\eta, \eta_1, \eta_0, \psi_x, \psi_\beta)$ evaluated at $\theta =\theta^*$. 
From \eqref{eq:ee_wlsx}, we have 
\beginy\label{eq:psistar_xreg}
\eta^* = \beginp  \eta_1^* \\  \eta_0^* \\  \psi_x^* \\ \psi_\beta^*\endp
\quad \text{with} \quad 
\begin{array}{lll}
\eta_1^*
&=& Z_i \cdot \rpn \cdot  \dyoa\tdxi   ,\\
\eta_0^* 
&=& (1-Z_i) \cdot \rpn \cdot   \dyza  \tdxi ,\\
\psi_x^* &=& x_i - \muxs, \\
\psi_\beta^* &=& U_i \left(\ryi - p_i\right).
\end{array}
\endy
The $B_\xreg^*$ matrix equals  
\begina
B_\xreg^* = \left.E\left(\eta \eta^\T\right)\right|_{\theta = \theta^*} = E\left(\eta^* \eta^{*\T}\right) 
= \left(\begin{array}{cc|cc|c|c}
d_{11} & d_{12} & 0 & 0 & d_{15} & d_{16} \\
d_{12}^\T & d_{22} & 0 & 0 & d_{25} & d_{26}  \\\hline
0 &0 & d_{33} & d_{34} & d_{35} & d_{36} \\
0 &0 & d_{34}^\T & d_{44} & d_{45} & d_{46} \\\hline
d_{15}^\T &d_{25}^\T & d_{35}^\T & d_{45}^\T & d_{55} & d_{56} \\\hline
d_{16}^\T &d_{26}^\T & d_{36}^\T & d_{46}^\T & d_{56}^\T & d_{66} \\
\end{array}\right)
=
\beginp
D_{11} & 0 & D_{13}&D_{14} \\
0 & D_{22} & D_{23}&D_{24} \\
D_{13}^\T & D_{23}^\T & D_{33} & D_{34}\\
D_{14}^\T & D_{24}^\T & D_{34}^\T & D_{44}
\endp,
\enda
where
\begine[(i)]
\item  
\begina
D_{11}= \beginp d_{11}& d_{12} \\ d_{12}^\T & d_{22} \endp &=& E\left(\eta_1^* \eta_1^{*\T}\right)\\
&=& E \left[ Z_i \cdot\dfrac{\riy}{ p_i^2} \cdot \dyoa^2 \cdot \tdxi  \tdxit  \right]\\
&=& E \left[ Z_i \cdot p_i^{-1} \cdot \dyoa^2 \cdot \tdxi  \tdxit  \right] 
\enda
with 
\begina
d_{11}  &=&  
 E\left[ p_i^{-1}  \cdot Z_i \cdot   \dyoa  ^2   \right] = b_{11}',\\
 d_{12}  
  &=&  E\left[   p_i^{-1} \cdot Z_i   \cdot \dyoa  ^2  \dxit \right], \\
d_{22} &=&  E\left[   p_i^{-1} \cdot Z_i   \cdot \dyoa  ^2 \dxi  \dxit\right] ,
\enda
\item \begina
D_{13} =  \beginp d_{15}\\ d_{25}  \endp &=& E\left(\eta_1^* \psi_x^{*\T}\right)\\
&=& E \left[ Z_i \cdot \rpn  \cdot \dyoa \tdxi  \cdot \dxit  \right]\\
&=& E \left[ Z_i \cdot    \dyoa    \cdot\tdxi \dxit  \right]\\
&=&  e   \cdot    \beginp  E \left[\dyoa  \cdot \dxit \right] \\ E \left[ \dyoa  \cdot \dxi   \dxit  \right]\endp  
\enda
with
\begina
d_{15}  
&=&  e \cdot \left( E\left[  \left\{Y_i(1)- \muos  \right\} \dxit \right]-   \gamma^{*\T}_1 E\left\{  \dxi   \dxit \right\} \right) = 0,\\
d_{25} &=&   e\cdot E\left[    \dyoa  \cdot \dxi \dxit \right],
\enda
\item 
\begina
D_{14} = \beginp d_{16}\\ d_{26}  \endp &=& E\left(\eta_1^* \psi_\beta^{*\T}\right)\\
&=& E \left[ Z_i \cdot \rpn   \cdot \dyoa   \tdxi \cdot (\ryi - p_i) \cdot U_i^\T  \right]\\
&=& E \left[ Z_i  \cdot \rpn (1-p_i) \cdot     \dyoa \tdxi  \cdot U_i^\T   \right]\\
&=&  E \left[ Z_i  \cdot  (1-p_i) \cdot       \dyoa \tdxi  \cdot U_i^\T   \right] 
\enda
with 
\begina
d_{16} &=& b_{13}',\\ 
d_{26} &=&E \left[ Z_i  \cdot  (1-p_i) \cdot       \dyoa (x_i - \muxs)  \cdot U_i^\T   \right]  ,
\enda 
\item $D_{22} = \beginp d_{33} & d_{34} \\ d_{34}^\T & d_{44}\endp$, $ D_{23} = \beginp
d_{35}\\ d_{45}\endp$, and $D_{24}= \beginp
d_{36}\\ d_{46}\endp$ with 
\begina
d_{33} = b_{22}', \quad d_{35} = 0, \quad  d_{36} = b_{23}' 
\enda 
by symmetry, 
\item 
\begina
   D_{33} &=& d_{55} = E(\psi_x^* \psi_x ^{*\T})  = \cov(x_i),\\
   D_{34} &=& d_{56} = E(\psi_x^* \psi_\beta^{*\T}) = E\left\{ \dxi (\ryi - p_i) U_i^\T \right\} = 0, \\
   D_{44} &=& d_{66} =  E(\psi_\beta^* \psi_\beta^{*\T}) = b_{33} \quad \text{by \eqref{eq:bs}}.
\enda
\ende 
This ensures
\beginy\label{eq:B_wlsx}
B_\xreg^* 
= \left(\begin{array}{cc|cc|c|c}
d_{11} & d_{12} & 0 & 0 & 0& d_{16} \\
d_{12}^\T & d_{22} & 0 & 0 & d_{25} & d_{26}  \\\hline
0 &0 & d_{33} & d_{34} & 0 & d_{36} \\
0 &0 & d_{34}^\T & d_{44} & d_{45} & d_{46} \\\hline
0&d_{25}^\T & 0 & d_{45}^\T & d_{55} & 0 \\\hline
d_{16}^\T &d_{26}^\T & d_{36}^\T & d_{46}^\T & 0 & d_{66} \\
\end{array}\right) 
\endy
with 
\begina
d_{11} = b_{11}',  \quad  d_{16} = b_{13}',\quad 
d_{33} = b_{22}',  \quad  d_{36} = b_{23}', \quad d_{66} = b_{33}.
\enda 
 
\bigskip
\noindent\underline{\textbf{Compute  $A_\xreg^*$.}}

\smallskip

With a slight abuse of notation, let $\pd{\mu_1} \eta_1^*$ denote the value of $\pd{\mu_1} \eta_1 $ evaluated at $\theta^*$. Similarly define other partial derivatives. 
Observe that 
\beginy\label{eq:derivative}
  \pd{\mux^\T}     \left\{Y_i(1)- (x_i - \mux)^\T  \gamma_1 - \mu_1 \right\} \beginp 1\\x_i - \mux\endp 
 &=& 
  \left\{Y_i(1)- (x_i - \mux)^\T  \gamma_1 - \mu_1 \right\} \beginp 
0 \\
 -I_J 
\endp +  \beginp 1\\x_i - \mux\endp \gamma_1^\T \nonumber\\
&=& 
\beginp 
\gamma_1^\T\\
 - \left\{Y_i(1)- (x_i - \mux)^\T  \gamma_1 - \mu_1 \right\}I_J + (x_i - \mux) \gamma_1^\T
\endp.\quad 
\endy
From \eqref{eq:ee_wlsx}, we have 
\begina
\left. \pd{(\mu_1, \gamma_1, \mu_0, \gamma_0, \mu_x, \beta)} \eta \right|_{\theta = \theta ^*} 
&=& 
\beginp
   \pd{(\mu_1^\T, \gamma_1^\T)} \eta_1^*    &  0   & \pd{\mu_x^\T} \eta_1^* &   \pd{\beta^\T} \eta_1 ^*  \\
0    & \pd{(\mu_0^\T, \gamma_0^\T)} \eta_0^*&  \pd{\mu_x^\T} \eta_0^* &   \pd{\beta^\T} \eta_0^* \\
0 & 0 & \pd{\mu_x^\T} \psi_x^* & 0\\
0 & 0 & 0 & \pd{\beta^\T} \psi_\beta^*
\endp,
\enda
where
\begina
   \pd{(\mu_1^\T, \gamma_1^\T)} \eta_1^* 
 &=&- Z_i \cdot\rpn \cdot \beginp
1\\
x_i - \muxs
\endp \tdxit,  \\
   \pd{\mu_x^\T} \eta_1^* 
 &=& Z_i \cdot\rpn \cdot \beginp 
\gamma_1^{*\T}\\
 - \dyoa I_J + \dxi  \gamma_1^{*\T}
\endp \quad \text{by \eqref{eq:derivative}}, \\
\pd{\beta^\T} \eta_1^*
 &=& Z_i \cdot\riy \cdot      \dyoa\tdxi \cdot \pd{\beta} p_i^{-1}\\
& =& - Z_i \cdot   \riy  \cdot    \dyoa \tdxi \cdot \dfrac{1-p_i}{ p_i }   U_i^\T \quad \text{by \eqref{eq:pdb}}\\
 &=&  - Z_i \cdot\rpn  (1-p_i)\cdot    \dyoa\tdxi  \cdot U_i^\T, \\
\pd{\mu_x^\T} \psi_x^* &=& -I_J.
\enda
Accordingly,  we have 
\begina
A_\xreg^* = -\left.E\left( \pd{(\mu_1, \gamma_1, \mu_0, \gamma_0, \mu_x, \beta)} \eta \right) \right|_{\theta = \theta^*}
&=&- E\left( \left. \pd{(\mu_1, \gamma_1, \mu_0, \gamma_0, \mu_x, \beta)} \eta \right|_{\theta = \theta ^*}\right) \\
&=& 
\left(\begin{array}{cc|cc|c|c}
g_{11} & g_{12} & 0 & 0 & g_{15} & g_{16} \\
g_{21} & g_{22} & 0 & 0 & g_{25} & g_{26}  \\\hline
0 &0 & g_{33} & g_{34} & g_{35} & g_{36} \\
0 &0 & g_{43} & g_{44} & g_{45} & g_{46} \\\hline
0 &0& 0 & 0 & g_{55} &0\\\hline
0& 0 & 0 & 0 & 0 & g_{66} \\
\end{array}\right),
\enda
where 
\begina
 \beginp g_{11} & g_{12} \\g_{21} & g_{22}\endp 
&=&- E \left(   \pd{(\mu_1^\T, \gamma_1^\T)} \eta_1^*   \right)\\
& =& E\left\{ Z_i \cdot\rpn \cdot \beginp
1\\
x_i - \muxs
\endp \tdxit 
\right\} =  E\left\{Z_i  \cdot \beginp
1\\
x_i - \muxs
\endp \tdxit 
\right\} \\
&=& e\cdot E\left\{ \beginp
1\\
x_i - \muxs
\endp \tdxit 
\right\} = e\cdot  \beginp
1 & 0 \\
0 & \cov(x_i) 
\endp,\\
\beginp g_{15}\\g_{25} \endp &=&   - E\left(  \pd{\mu_x^\T} \eta_1 ^*\right) \\
& =&  - E \left\{ Z_i \cdot\rpn \cdot \beginp 
\gamma_1^{*\T}\\
 - \left\{Y_i(1)- (x_i - \mux)^\T  \gamma^*_1 - \muos  \right\}I + \dxi  \gamma_1^{*\T}
\endp\right\}\\
& =& - E \left\{ Z_i   \cdot \beginp 
\gamma_1^{*\T}\\
 - \left\{Y_i(1)- (x_i - \mux)^\T  \gamma^*_1 - \muos  \right\}I_J + \dxi  \gamma_1^{*\T}
\endp\right\}\\
& =& - e\cdot  \beginp 
\gamma_1^{*\T}\\
 0_{J\times J}
\endp, \\
\beginp g_{16}\\g_{26} \endp  &=& - E\left(  \pd{\beta^\T} \eta_1 ^* \right) = E \left(  Z_i \cdot\rpn  (1-p_i) \cdot \dyoa\cdot  \tdxi   \cdot U_i^\T\right)\\
&=& E \left[   Z_i \cdot   (1-p_i) \cdot \dyoa \cdot  \tdxi  \cdot U_i^\T\right]\\
&=& D_{14},\\\\
 \beginp g_{33} & g_{34} \\g_{43} & g_{44}\endp  &=&   (1-e)\cdot  \beginp
1 & 0 \\
0 & \cov(x_i)\endp, \qquad   \beginp g_{35}\\g_{45} \endp=- (1- e)\cdot  \beginp 
\gamma_0^{*\T}\\
 0_{J\times J}
\endp,
\qquad \beginp g_{36}\\g_{46} \endp = D_{24} \quad \text{by symmetry},
\enda
 and 
\begina
g_{55} = - E \left(\pd{\mu_x^\T} \psi_x^*  \right) = I_J, \qquad g_{66} = - E \left(\pd{\mu_\beta^\T} \psi_\beta^*  \right)  =  d_{66}.
\enda
 This ensures 
\begina
A_\xreg^* = \left(\begin{array}{cccc|cc}
g_{11} & 0 & 0 & 0 & g_{15} & d_{16} \\
0 & g_{22} & 0 & 0 & 0& d_{26}  \\
0 &0 & g_{33} & 0 & g_{35} & d_{36} \\
0 &0 & 0& g_{44} & 0 & d_{46} \\\hline
0 &0& 0 & 0 & I &0\\
0& 0 & 0 & 0 & 0 & d_{66} \\
\end{array}\right)  = \left(\begin{array}{c|c}
G_{11} & G_{12} \\\hline
0 & G_{22}
\end{array}\right),
\enda
where $
g_{11} = e$, $g_{33} = 1-e$, $g_{15} = -e\gost$, and $g_{35} = -(1-e) \gzst$. 
By the formula of block matrix inverse, we have
\begina
A_\xreg^{*-1} 
=  \beginp
G_{11}^{-1} &  -G_{11}^{-1}G_{12} G_{22}^{-1} \\
0 & G_{22}^{-1}
\endp
= \left(\begin{array}{cccc|cc}
g^{-1}_{11} & 0 & 0 & 0 & -g^{-1}_{11} g_{15}  & -g^{-1}_{11} d_{16}d^{-1}_{66} \\
0 & g^{-1}_{22} & 0 & 0 & 0& -g^{-1}_{22} d_{26} d^{-1}_{66} \\
0 &0 & g^{-1}_{33} & 0 & -g^{-1}_{33} g_{35}& -g^{-1}_{33} d_{36} d^{-1}_{66}\\
0 &0 & 0& g^{-1}_{44} & 0 & -g^{-1}_{44}d_{46}d^{-1}_{66} \\\hline
0 &0& 0 & 0 & I &0\\
0& 0 & 0 & 0 & 0 & d^{-1}_{66} \\
\end{array}\right)  
\enda
with
\beginy\label{eq:A_wlsx}
 \beginp
1 & 0 & 0 & 0 & 0 &0\\
0 & 0 & 1 & 0 & 0&0 
\endp
A_\xreg^{*-1} &=& 
\beginp
g^{-1}_{11} & 0 & 0 & 0 & -g^{-1}_{11} g_{15}  & -g^{-1}_{11} d_{16}d^{-1}_{66} \\
0 &0 & g^{-1}_{33} & 0 & -g^{-1}_{33} g_{35}  & -g^{-1}_{33} d_{36} d^{-1}_{66}
\endp \nonumber\\
&=&
\beginp
g_{11}^{-1} & 0 \\
0 & g_{33}^{-1}
\endp 
\beginp
1 & 0 & 0 & 0 & - g_{15}  & -d_{16}d^{-1}_{66} \\
0 &0 & 1 & 0 & -g_{35}  & -d_{36} d^{-1}_{66}
\endp,
\endy
where $
g_{11} = e$, $g_{33} = 1-e$, $g_{15} = -e\gost$, and $g_{35} = -(1-e) \gzst$.

\bigskip

\noindent\underline{\textbf{Compute $v_\xreg$.}}

\smallskip 

Direct algebra ensures
\beginy\label{eq:algebra}
&&\left( \begin{array}{cccc|cc}
1 & 0 & 0 & 0 & - g_{15}  & -d_{16}d^{-1}_{66} \\
0 &0 & 1 & 0 & -g_{35} & -d_{36} d^{-1}_{66}
\end{array}\right)
\left(\begin{array}{cccc|cc}
d_{11} & d_{12} & 0 & 0 & 0 & d_{16} \\
d_{12}^\T & d_{22} & 0 & 0 & d_{25} & d_{26}  \\
0 &0 & d_{33} & d_{34} & 0& d_{36} \\
0 &0 & d_{34}^\T & d_{44} & d_{45} & d_{46} \\\hline
0 &d_{25}^\T & 0 & d_{45}^\T & d_{55} &0 \\
d_{16}^\T &d_{26}^\T & d_{36}^\T & d_{46}^\T &0 & d_{66} \\
\end{array}\right)
\beginp
1 & 0 \\
 0 & 0 \\
 0 & 1 \\
 0 & 0 \\\hline
 -  g_{15}^\T& - g_{35} ^\T\\
 - d^{-1}_{66}d_{16}^\T & - d^{-1}_{66}d_{36}^\T
\endp\nonumber\\
&=&\beginp
1 & 0 & 0 & 0  \\
0 &0 & 1 & 0 
\endp
\left(\begin{array}{cccc}
d_{11} & d_{12} & 0 & 0  \\
d_{12}^\T & d_{22} & 0 & 0   \\
0 &0 & d_{33} & d_{34}  \\
0 &0 & d_{34}^\T & d_{44} 
\end{array}\right)
\beginp
1 & 0 \\
 0 & 0 \\
 0 & 1 \\
 0 & 0
\endp+\beginp
g_{15} & d_{16}d^{-1}_{66} \\
 g_{35} & d_{36} d^{-1}_{66}
\endp
\beginp
 d_{55} & 0 \\
  0 & d_{66} \\
\endp 
\beginp
 g_{15}^\T& g_{35} ^\T\\
  d^{-1}_{66}d_{16}^\T & d^{-1}_{66}d_{36}^\T
\endp \nonumber\\
&&- \beginp
g_{15} & d_{16}d^{-1}_{66} \\
 g_{35} &d_{36} d^{-1}_{66}
\endp
\beginp 
0 &d_{25}^\T & 0 & d_{45}^\T  \\
d_{16}^\T &d_{26}^\T & d_{36}^\T & d_{46}^\T
\endp 
\beginp
1 & 0 \\
 0 & 0 \\
 0 & 1 \\
 0 & 0
\endp-\beginp
1 & 0 & 0 & 0 \\
0 &0 & 1 & 0 
\endp
\left(\begin{array}{cc}
 0 & d_{16} \\
 d_{25} & d_{26}  \\
0& d_{36} \\
 d_{45} & d_{46} 
\end{array}\right)
\beginp
   g_{15}^\T&  g_{35} ^\T\\
 d^{-1}_{66}d_{16}^\T &  d^{-1}_{66}d_{36}^\T
\endp\nonumber\\
&=& \beginp d_{11} & 0 \\ 0 & d_{33}\endp + \beginp g_{15} \\ g_{35} \endp  d_{55} \beginp g_{15}^\T, & g_{35}^\T \endp +
\beginp
  d_{16}  \\
 d_{36} 
\endp
d^{-1}_{66}
\beginp
d_{16}^\T, & d_{36}^\T
\endp \nonumber\\
&& - \beginp
  d_{16}  \\
 d_{36} 
\endp
d^{-1}_{66}
\beginp
d_{16}^\T, & d_{36}^\T
\endp - \beginp
  d_{16}  \\
 d_{36} 
\endp
d^{-1}_{66}
\beginp
d_{16}^\T, & d_{36}^\T
\endp\nonumber\\
&=& \beginp d_{11} & 0 \\ 0 & d_{33}\endp + \beginp g_{15} \\ g_{35} \endp  d_{55} \beginp g_{15}^\T, g_{35}^\T \endp -
\beginp
  d_{16}  \\
 d_{36} 
\endp
d^{-1}_{66}
\beginp
d_{16}^\T & d_{36}^\T
\endp, 
\endy
where
\begina d_{11} = b_{11}', \quad d_{33} = b_{22}', \quad g_{15} = -e\gost, \quad g_{35} = -(1-e) \gzst, \quad d_{55} = \cov(x_i),\quad d_{16} = b_{13}', \quad d_{36} = b_{23}', \quad d_{66} = b_{33}
\enda 
from previous computations. 
Equations \eqref{eq:clt_wlsx},  \eqref{eq:B_wlsx},  \eqref{eq:A_wlsx}, and \eqref{eq:algebra} together ensure   
\begina
V_\xreg &=& \cova\left\{
\beginp
\hy_\xreg(1) \\
\hy_\xreg(0)  
\endp 
\right\}
=
\cova\left\{
\beginp
1 & 0 & 0 & 0 & 0 &0\\
0 & 0 & 1 & 0 & 0&0 
\endp
\beginp\hy_\xreg(1) \\ \hgo \\ \hy_\xreg(0) \\ \hgz  \\ \bar x \\ \hat\beta\endp
\right\}
 \\
&\overset{\eqref{eq:clt_wlsx}}{=}& \beginp
1 & 0 & 0 & 0 & 0 &0\\
0 & 0 & 1 & 0 & 0&0 
\endp
\abareg 
\beginp
1 & 0 \\
 0 & 0 \\
 0&1\\
 0 &0\\
0 & 0 \\
 0&0 
\endp\\
&\overset{\eqref{eq:A_wlsx}}{=}&\beginp 
g_{11}^{-1} & 0 \\
0 & g_{33}^{-1}
\endp \cdot \\
&&\beginp
1 & 0 & 0 & 0 & - g_{15}  & -d_{16}d^{-1}_{66} \\
0 &0 & 1 & 0 & -g_{35}  & -d_{36} d^{-1}_{66}
\endp 
\left(\begin{array}{cccc|cc}
d_{11} & d_{12} & 0 & 0 & 0 & d_{16} \\
d_{12}^\T & d_{22} & 0 & 0 & d_{25} & d_{26}  \\
0 &0 & d_{33} & d_{34} & 0& d_{36} \\
0 &0 & d_{34}^\T & d_{44} & d_{45} & d_{46} \\\hline
0 &d_{25}^\T & 0 & d_{45}^\T & d_{55} & 0 \\
d_{16}^\T &d_{26}^\T & d_{36}^\T & d_{46}^\T &0 & d_{66} \\
\end{array}\right)
\beginp
1 & 0 \\
 0 & 0 \\
 0 & 1 \\
 0 & 0 \\
 -  g_{15}^\T& -  g_{35} ^\T\\
 - d^{-1}_{66}d_{16}^\T & - d^{-1}_{66}d_{36}^\T
\endp\\
&&\beginp
g_{11}^{-1} & 0 \\
0 & g_{33}^{-1}
\endp\\
&\overset{\eqref{eq:algebra}}{=}&\beginp 
g_{11}^{-1} & 0 \\
0 & g_{33}^{-1}
\endp \cdot \left\{ \beginp d_{11} & 0 \\ 0 & d_{33}\endp + \beginp g_{15} \\ g_{35} \endp  d_{55} \beginp g_{15}^\T, g_{35}^\T \endp -
\beginp
  d_{16}  \\
 d_{36} 
\endp
d^{-1}_{66}
\beginp
d_{16}^\T & d_{36}^\T
\endp \right\} \beginp
g_{11}^{-1} & 0 \\
0 & g_{33}^{-1}
\endp\\
&=&    \beginp \dfrac{b'_{11}}{e^2} & 0 \\ 0 & \dfrac{b'_{22}}{(1-e)^2} \endp + \beginp \gost  \\ \gzst \endp  \cov(x_i) \beginp \gos, \gzs \endp -
\beginp
 \dfrac{ b'_{13}}{e}  \\
 \dfrac{ b'_{23}}{1-e}
\endp
b^{-1}_{33}
\beginp
\dfrac{(b'_{13})^\T}{e} &  \dfrac{ (b'_{23})^\T}{1-e}
\endp. 
\enda
This ensures
\begina
v_\xreg = \va(\htreg ) &=& \va\left\{\hywlsxo - \hywlsxz\right\}=  \va\left\{(1, \ -1) \beginp \hywlsxo\\ \hywlsxz \endp \right\} \\
&=& (1, \ -1) V_\xreg \beginp 1\\-1\endp \\
&=& \dfrac{b'_{11}}{e^2} + \dfrac{b'_{22}}{(1-e)^2} + (\gos - \gzs)^\T\cov(x_i) (\gos - \gzs) - \left(  \dfrac{ b'_{13}}{e} -
 \dfrac{ b'_{23}}{1-e} \right)d_{66}^{-1} \left(  \dfrac{ b'_{13}}{e} -
 \dfrac{ b'_{23}}{1-e} \right)^\T. 
\enda
\end{proof}

\subsection{Proof of Theorem \ref{thm:clt_app} \eqref{item:complete}--\eqref{item:linear}}\label{sec:thm1_special cases}

Direct algebra shows that 
\beginy\label{eq:vv}
\var\{Y_i'(z)\}  = \var\{\yizz - x_i^\T\gamma_z\} 
&=& \var\{\yizz\} +\var(x_i^\T\gamma_z) - 2\cov\{\yizz, x_i^\T\gamma_z\}\nonumber\\
&=& \var\{\yizz\}  +\gamma_z^\T\cov(x_i)\gamma_z - 2\gamma_z^\T \cov(x_i) \gamma_z\nonumber\\
& =& \var\{\yizz\}  -\gamma_z^\T\cov(x_i)\gamma_z \quad(z=0,1).  
\endy
In addition, it follows from 
$
\cov(x_i, \tyi) = \cov \{x_i, e^{-1}Y_i(1) + (1-e)^{-1}Y_i(0)\} =  \cov(x_i) \{e^{-1}\gamma_1 +   (1-e)^{-1}\gamma_0\}$
that
\begina 
\var\left\{\proj(\tyi \mid 1, x_i)\right\} &=& \cov(\tyi, x_i) \{\cov(x_i)\}^{-1} \cov(x_i, \tyi)\\
&=& \left\{e^{-1}\gamma_1^\T +   (1-e)^{-1}\gamma_0^\T\right\} \cov(x_i) \left\{e^{-1}\gamma_1 +   (1-e)^{-1}\gamma_0\right\}\\
&=& e^{-2}\gamma_1^\T \cov(x_i) \gamma_1 + (1-e)^{-2} \gamma_0^\T \cov(x_i) \gamma_0 + e^{-1}(1-e)^{-1} \left\{ \gamma_1^\T \cov(x_i) \gamma_0 +  \gamma_0^\T \cov(x_i) \gamma_1\right\}.
\enda
This ensures
\beginy\label{eq:d_xps}
v_\ua - v_\xps &=& e(1-e) \var\left\{\proj(\tyi \mid 1, x_i)\right\}\nonumber\\
&=& \frac{1-e}{e} \gamma_1^\T \cov(x_i) \gamma_1 + \frac{e}{1-e}\gamma_0^\T \cov(x_i) \gamma_0 + \gamma_1^\T \cov(x_i) \gamma_0 +  \gamma_0^\T \cov(x_i) \gamma_1\nonumber\\
&=& \left( \frac{1}{e} - 1\right)  \gamma_1^\T \cov(x_i) \gamma_1 + \left(\frac{1}{1-e} - 1\right)\gamma_0^\T \cov(x_i) \gamma_0 + \gamma_1^\T \cov(x_i) \gamma_0 +  \gamma_0^\T \cov(x_i) \gamma_1 \nonumber\\
&=&  \frac{\gamma_1^\T \cov(x_i) \gamma_1}{e}    + \frac{\gamma_0^\T \cov(x_i) \gamma_0}{1-e}  -( \gamma_1 - \gamma_0)^\T \cov(x_i)(\gamma_1- \gamma_0).
\endy

\subsubsection{Proof of Theorem \ref{thm:clt_app}\eqref{item:complete}}
When $p_i = 1$, we have 
\renewcommand{\arraystretch}{1.2}
\begina
\begin{array}{lll}
b_{11} = 
 e \var\{\yio\},&&
   b_{11}' = e\var\{Y_i'(1)\}  ,\\
   b_{22} = (1-e) \var\{\yiz\}, && 
  b_{22}' =(1-e)  \var\{Y_i'(0)\},
 \end{array}\enda
while $ b_{13} = b_{23} = b_{13}' = b_{23}' = 0_m$ and $b_{33} = b_{33}' = 0_{m\times m}$, where $m$ denotes the dimension of $U_i$. 
This, together with \eqref{eq:vv}, verifies the simplified expression of $v_\ua$ and ensures 
\begina
v_\ua - v_\xreg &=& \frac{b_{11}-b_{11}'}{e^2} +  \frac{b_{22}-b_{22}'}{(1-e)^2} - (\gamma_1-\gamma_0)^\T\cov(x_i)(\gamma_1-\gamma_0)\\
&\overset{\eqref{eq:vv}}{ =}& e^{-1} \gamma_1^\T \cov(x_i) \gamma_1 + (1-e)^{-1}\gamma_0^\T \cov(x_i) \gamma_0 -  (\gamma_1-\gamma_0)^\T\cov(x_i)(\gamma_1-\gamma_0)\\
&\overset{\eqref{eq:d_xps}}{ =}& v_\ua - v_\ipwx.
\enda

\subsubsection{Proof of Theorem \ref{thm:clt_app}\eqref{item:mcar}}
Given $p_i = p \in (0,1)$ for all $i$, we have $U_i = 1$ and hence 
\begina
b_{11} = p^{-1} e \cdot \var\{Y_i(1)\}, \quad b_{22} = p^{-1} (1-e) \cdot \var\{Y_i(0)\}, \quad b_{13} = b_{23} = 0,\\
b'_{11} = p^{-1} e \cdot \var\{Y'_i(1)\}, \quad b_{22} = p^{-1} (1-e) \cdot \var\{Y'_i(0)\}, \quad b'_{13} = b'_{23} = 0.
\enda
This, together with \eqref{eq:vv}, verifies the simplified expressions of $v_\ua$ and $v_\xreg$ with
\beginy\label{eq:d_xreg}
v_\ua - v_\xreg =  \frac{\gamma_1^\T \cov(x_i) \gamma_1}{ep} +  \frac{\gamma_0^\T \cov(x_i) \gamma_0}{(1-e)p} - (\go-\gz)^\T\cov(x_i)(\go-\gz). 
\endy
Compare \eqref{eq:d_xreg} with \eqref{eq:d_xps} to see
\begina
v_\xreg - v_\xps &=& (v_\ua - v_\xps) - (v_\ua - v_\xreg)\\
&=&  \frac{1}{e}\left( 1 - \frac{1}{p}\right)  \gamma_1^\T \cov(x_i) \gamma_1 + \frac{1}{1-e}\left( 1- \frac{1}{p}\right)\gamma_0^\T \cov(x_i) \gamma_0\\
&=& (1-p^{-1})\left\{ \frac{\gamma_1^\T \cov(x_i) \gamma_1}{e}   + \frac{\gamma_0^\T \cov(x_i) \gamma_0}{1-e}  \right\} \leq 0.
\enda

\bigskip

\subsubsection{Proof of Theorem \ref{thm:clt_app}\eqref{item:linear}, $v_\xps - v_\xreg   = E(\Gamma^\T\Gamma)$}
We show in this subsection that $v_\xps - v_\xreg   = E(\Gamma^\T\Gamma)$. The result ensures the relative order between $\{\htn, \htreg, \hteps\}$ with  $v_\xreg \leq v_\xps \leq v_\ua$. 
The proof of $v_\ua - v_\xreg   = \va(\htn - \htreg)$ is long, and we relegate it to Section \ref{sec:long}.

With $E\{Y_i(z) \mid x_i\} = \mu_z + (x_i - \mu_x)^\T\gamma_z$, 
we have 
\beginy\label{eq:linear_0}
&&E\{Y_i(z) - \mu_z \mid x_i, Z_i\} =E\{Y_i(z) - \mu_z \mid x_i\}= (x_i - \mu_x)^\T\gamma_z,\nonumber\\ 
&& 
E\{Y_i'(z)- \mu_z' \mid x_i, Z_i\} = E\{Y_i'(z) - \mu_z' \mid x_i\}=E\{Y_i(z) - x_i^\T\gamma_z - (\mu_z - \mu_x^\T\gamma_z)  \mid x_i\} = 0.
\endy
This ensures 
\beginy\label{eq:b13}
b_{13} 
  &=& E  \left[ (1-p_i) \cdot Z_i \cdot \dyon  \cdot U_i^\T  \right] \nonumber \\
  &=&
  E  \left[ (1-p_i) \cdot Z_i \cdot E\left\{\yio - \mu_1\mid x_i, Z_i\right\}  \cdot U_i^\T  \right] = \gamma_1^\T E  \big\{ (1-p_i) \cdot Z_i \cdot (x_i - \mu_x) \cdot U_i^\T  \big\},\\
  b'_{13} 
  &=& E  \left[ (1-p_i) \cdot Z_i \cdot \{Y_i'(1) - \mu_1'\}  \cdot U_i^\T  \right]\nonumber \\
  &=&
  E  \left[ (1-p_i) \cdot Z_i \cdot E\left\{Y_i'(1) - \mu_1' \mid x_i, Z_i\right\}  \cdot U_i^\T  \right] = 0,\nonumber
\endy
and by symmetry
\beginy\label{eq:b23}
b_{23} 
 =  \gamma_0^\T E  \big\{ (1-p_i) \cdot (1-Z_i) \cdot (x_i - \mu_x) \cdot U_i^\T  \big\},\quad
  b'_{23} 
  = 0.
\endy
Consequently,  the expression of $v_\xreg$ simplifies to 
\begina
v_\xreg =  \dfrac{b'_{11}}{e^2} + \dfrac{b'_{22}}{(1-e)^2} +(\go-\gz)^\T\cov(x_i)(\go-\gz)
\enda
with
\beginy\label{eq:d_xreg_linear}
v_\ua -v_\xreg = \dfrac{b_{11}-b_{11}'}{e^2} + \dfrac{b_{22} - b_{22}'}{(1-e)^2} - \left(\dfrac{b_{13}}{e}-\dfrac{b_{23}}{1-e}\right)b_{33}^{-1}\left(\dfrac{b_{13}}{e}-\dfrac{b_{23}}{1-e}\right)^\T - (\go-\gz)^\T\cov(x_i)(\go-\gz).
\endy
Compare \eqref{eq:d_xreg_linear} with \eqref{eq:d_xps} to see
\beginy\label{eq:reg_ps}
&&  v_\xreg - v_\xps  \nonumber\\
&=& (v_\ua - v_\xps) - (v_\ua - v_\xreg) \nonumber\\
&=& \frac{\gamma_1^\T \cov(x_i) \gamma_1}{e}    + \frac{\gamma_0^\T \cov(x_i) \gamma_0}{1-e} - ( \gamma_1 - \gamma_0)^\T \cov(x_i)(\gamma_1- \gamma_0)\nonumber\\
&& -  \left\{\dfrac{b_{11}-b_{11}'}{e^2} + \dfrac{b_{22} - b_{22}'}{(1-e)^2} - \left(\dfrac{b_{13}}{e}-\dfrac{b_{23}}{1-e}\right)b_{33}^{-1}\left(\dfrac{b_{13}}{e}-\dfrac{b_{23}}{1-e}\right)^\T - (\go-\gz)^\T\cov(x_i)(\go-\gz)\right\}\nonumber\\
&=&\left(\dfrac{b_{13}}{e}-\dfrac{b_{23}}{1-e}\right)b_{33}^{-1}\left(\dfrac{b_{13}}{e}-\dfrac{b_{23}}{1-e}\right)^\T - \left\{  \dfrac{b_{11}-b_{11}'}{e^2} - \frac{\gamma_1^\T \cov(x_i) \gamma_1}{e} +
  \dfrac{b_{22} - b_{22}'}{(1-e)^2} - \frac{\gamma_0^\T \cov(x_i) \gamma_0}{1-e} \right\} .
\endy
We show in the following that the right-hand side of \eqref{eq:reg_ps} equals $-E(\Gamma^\T\Gamma)$.

First,  recall that 
\begina
A = \sqrt{p_i(1-p_i)} U_i^\T, \quad B = \sqrt{\dfrac{1-p_i}{p_i}} (x_i-\mu_x)^\T  \left( \frac{Z_i}{e}  \gamma_1- \frac{1-Z_i}{1-e} \gamma_0\right)
\enda
with $b_{33} = E\{p_i(1-p_i) \cdot U_i U_i^\T\} = E(A^\T A)$. 
It follows from \eqref{eq:b13} and \eqref{eq:b23} that 
\begina
\dfrac{b_{13}}{e}-\dfrac{b_{23}}{1-e} &=& 
\dfrac{ \gamma_1^\T E  \big\{ (1-p_i) \cdot Z_i \cdot (x_i - \mu_x) \cdot U_i^\T  \big\}}{e}-\dfrac{\gamma_0^\T E  \big\{ (1-p_i) \cdot (1-Z_i) \cdot (x_i - \mu_x) \cdot U_i^\T  \big\}}{1-e}\\
&=& E\left\{(1-p_i) \left( \frac{Z_i}{e}  \gamma_1- \frac{1-Z_i}{1-e} \gamma_0\right)^\T(x_i - \mu_x) \cdot U_i^\T  \right\}\\
&=& E(B^\T A).
\enda
This ensures
\beginy\label{eq:bb_2}
\left(\dfrac{b_{13}}{e}-\dfrac{b_{23}}{1-e}\right)b_{33}^{-1}\left(\dfrac{b_{13}}{e}-\dfrac{b_{23}}{1-e}\right)^\T
= E(B^\T A)\{E(A^\T A)\}^{-1} E(A^\T B) = E(B^\T B)-E(\Gamma^\T\Gamma)
\endy 
by Lemma \ref{lem:cs}.

Next,  it follows from the definition of $Y_i'(1)$ and $\mu_1'$ that
$
 \yio - \mu_1= Y_i'(1) - \mu_1' + (x_i-\mu_x)^\T\gamma_1
$ such that
\begina
\dyon^2 - \{Y_i'(1) - \mu_1'\}^2 
&=& \left[ \{Y_i'(1) - \mu_1'\} + (x_i-\mu_x)^\T\gamma_1\right]^2 - \{Y_i'(1) - \mu_1'\}^2  \\
&=& 2\{Y_i'(1) - \mu_1'\} (x_i-\mu_x)^\T\gamma_1 + \gamma_1^\T (x_i-\mu_x)(x_i-\mu_x)^\T\gamma_1
\enda
with 
\begina
E\big[ \dyon^2 - \{Y_i'(1) - \mu_1'\}^2  \mid x_i, Z_i\big] 
&=& 2 E\big\{Y_i'(1) - \mu_1'  \mid x_i, Z_i \big\} (x_i-\mu_x)^\T\gamma_1 +\gamma_1^\T (x_i-\mu_x)(x_i-\mu_x)^\T\gamma_1 \\
&\overset{\eqref{eq:linear_0}}{=}& \gamma_1^\T (x_i-\mu_x)(x_i-\mu_x)^\T\gamma_1 
\enda
from \eqref{eq:linear_0}. This ensures
\beginy\label{eq:bb}
b_{11}-b_{11}' &=& E \Big( p_i^{-1}\cdot Z_i \cdot \left[ \dyon^2 - \{Y_i'(1) - \mu_1'\}^2  \right] \Big) \nonumber\\
&=& E \Big( p_i^{-1}\cdot Z_i \cdot E\left[ \dyon^2 - \{Y_i'(1) - \mu_1'\}^2 \mid x_i, Z_i \right] \Big) \nonumber\\
&=& E \left\{ p_i^{-1}\cdot Z_i \cdot  \gamma_1^\T (x_i-\mu_x)(x_i-\mu_x)^\T\gamma_1 \right\}
\endy
and by symmetry 
\beginy\label{eq:bb_}
b_{22}-b_{22}' =  E \left\{ p_i^{-1}\cdot (1-Z_i) \cdot  \gamma_0^\T(x_i-\mu_x)(x_i-\mu_x)^\T\gamma_0\right\}.
\endy
Using the definition of $B$ and the basic fact that $Z_i(1-Z_i)=0$, we have  
\begina
B^\T B 
&=& \dfrac{1-p_i}{p_i} \cdot \left( \frac{Z_i}{e}  \gamma_1^\T - \frac{1-Z_i}{1-e} \gamma_0^\T\right) (x_i-\mu_x)(x_i-\mu_x)^\T \left( \frac{Z_i}{e}  \gamma_1  - \frac{1-Z_i}{1-e} \gamma_0 \right)\\
&=& \dfrac{1-p_i}{p_i} \cdot \left( \frac{Z_i}{e}\right)^2  \gamma_1^\T  (x_i-\mu_x)(x_i-\mu_x)^\T   \gamma_1   + \dfrac{1-p_i}{p_i} \cdot \left(  \frac{1-Z_i}{1-e} \right)^2  \gamma_0^\T  (x_i-\mu_x)(x_i-\mu_x)^\T   \gamma_0 \\
&& +\dfrac{1-p_i}{p_i} \cdot   \frac{Z_i}{e} \cdot  \frac{1-Z_i}{1-e} \cdot \Big\{\gamma_1^\T (x_i-\mu_x)(x_i-\mu_x)^\T  \gamma_0  + \gamma_0^\T (x_i-\mu_x)(x_i-\mu_x)^\T   \gamma_1\Big\} \\
&=& \dfrac{1-p_i}{p_i} \cdot  \frac{Z_i}{e^2} \cdot   \gamma_1^\T  (x_i-\mu_x)(x_i-\mu_x)^\T   \gamma_1   + \dfrac{1-p_i}{p_i} \cdot   \frac{1-Z_i}{(1-e)^2} \cdot  \gamma_0^\T  (x_i-\mu_x)(x_i-\mu_x)^\T   \gamma_0\\
&=&    \frac{p_i^{-1}Z_i}{e^2} \cdot   \gamma_1^\T  (x_i-\mu_x)(x_i-\mu_x)^\T   \gamma_1 - \frac{Z_i}{e^2} \cdot   \gamma_1^\T  (x_i-\mu_x)(x_i-\mu_x)^\T   \gamma_1  \\
&&+    \frac{p_i^{-1}(1-Z_i)}{(1-e)^2} \cdot  \gamma_0^\T  (x_i-\mu_x)(x_i-\mu_x)^\T   \gamma_0 -    \frac{1-Z_i}{(1-e)^2} \cdot  \gamma_0^\T  (x_i-\mu_x)(x_i-\mu_x)^\T   \gamma_0.
\enda
This, together with \eqref{eq:bb} and \eqref{eq:bb_},   ensures
\beginy\label{eq:bb_1}
E(B^\T B) = \frac{b_{11}-b_{11}'}{e^2} - \frac{\gamma_1^\T  \cov(x_i)   \gamma_1}{e} + \frac{b_{22} - b_{22}'}{(1-e)^2} - \frac{\gamma_0^\T \cov(x_i) \gamma_0}{1-e}. 
\endy
Plugging  \eqref{eq:bb_2} and \eqref{eq:bb_1} in \eqref{eq:reg_ps} verifies $v_\xreg - v_\xps=-E(\Gamma^\T\Gamma)$.

\subsubsection{Proof of Theorem \ref{thm:clt_app}\eqref{item:linear}, $  v_\ua - v_\xreg=  \va(\htn - \htreg)$}\label{sec:long}
We give in this subsection a direct proof of $  v_\ua - v_\xreg=  \va(\htn - \htreg)$ without invoking the semiparametric theory.

Define
\beginy\label{eq:dx}
\hdo  = \dfrac{\sumi Z_i \cdot \rp \cdot  (x_i - \bar x)  }{\sumi Z_i \cdot \rp}, \qquad 
\hdz   = \dfrac{\sumi (1-Z_i) \cdot \rp \cdot  (x_i - \bar x)  }{\sumi (1-Z_i) \cdot \rp}
\endy
to write 
\begina
\hywlsxo = \hyipwo - \hdo^\T \hat \gamma_1, \qquad 
\hywlsxz = \hyipwz - \hdz^\T\hat\gamma_0
 \enda
 by Proposition \ref{prop:hajek_app} and Lemma \ref{lem:xreg_num}.
This ensures
\begina
\ \htn  - \htreg  = \Delta, \quad \text{where}\quad \Delta = 
 \hdo ^\T \hgo - \hdz  ^\T \hgz. 
\enda
Then 
\beginy\label{eq:v_wls_decomp}
\va(\htn ) = \va(\htreg ) + 2  \cova(\htreg , \Delta) + \va(\Delta).
\endy
To verify the result, it suffices to show that $\cova(\htreg , \Delta) = 0$. 
We do this in the following three steps. 
\begine[(i)]
\item\label{step:2} Show that the asymptotic distribution of $\sqrt N( \htreg ,   \Delta) $ equals that of 
$\sqrt N ( \htau'_\xreg,  \Delta') $, 
where $( \htau'_\xreg,  \Delta')$ are variants of $( \htreg ,  \Delta) $ with $(\hgo, \hgz)$ replaced by their probability limits $(\gos, \gzs)$. That is, 
\begina
\htau'_\xreg = \hywlsxop  - \hywlsxzp
\enda
with 
$
\hywlsxop  = \hyipwo - \hdo^\T \gos$ and $\hywlsxzp  = \hyipwz - \hdz^\T \gzs$, and 
\begina
\Delta' = \hdo^\T\gos - \hdz^\T \gzs. 
\enda 
 
\item\label{step:3} Compute the asymptotic joint distribution of $(\hywlsxop , \hdo , \hywlsxzp , \hdz )$ by m-estimation. 
\item\label{step:4} Compute $\cova(\htreg , \Delta)$ 
based on 
\beginy\label{eq:true_cov}
\cova(\htreg , \Delta) &=& \cova( \htau'_\xreg, \ \Delta') \nonumber \\
&=& \cova\left\{\hywlsxop  - \hywlsxzp , \ \hdo ^\T \gos - \hdz  ^\T \gzs \right\} \nonumber\\
&=& \cova\left\{\hywlsxop , \ \hdo  \right\} \gos -\cova\left\{\hywlsxop   , \  \hdz \right\}\gzs \nonumber \\
&&- \cova\left\{ \hywlsxzp , \ \hdo   \right\} \gos+\cova\left\{ \hywlsxzp , \ \hdz   \right\}\gzs.
\endy
\ende
 
\noindent\textbf{\underline{Step \eqref{step:2}.}}

By Slutsky's theorem,   it suffices to show that
\beginy
\sqrt N
\left\{\beginp 
\htreg '   \\
\Delta'  
\endp - 
\beginp 
 \htreg \\
 \Delta
\endp \right\} = \op. \label{eq:step1}
\endy
To this end, observe that 
\begina
\htreg ' + \Delta' = \htn  = \htreg  + \Delta.  
\enda
This ensures 
\begina
\htreg ' - \htreg  = - ( \Delta' - \Delta) = \hdo^\T(\hgo-\gos) - \hdz^\T (\hgz - \gzs) 
\enda
with 
\begina
\sqrt N
\left\{\beginp 
\htreg '   \\
\Delta'  
\endp - 
\beginp 
 \htreg \\
 \Delta
\endp \right\} = \beginp 1\\-1\endp \left\{ \sqrt N \hdo^\T(\hgo-\gos) - \sqrt N\hdz^\T (\hgz - \gzs)  \right\}.
\enda
With $\hgo - \gos = \op$ and $\hgz -\gzs = \op$ from \eqref{eq:clt_wlsx}, by Slutsky's theorem, a sufficient condition for \eqref{eq:step1} is that  $\sqrt N \hdo $ and $\sqrt N \hdz $ are both asymptotically   normal with mean zero. We verify this below by computing the asymptotic joint distribution of $\sqrt N \hdo $ and $\sqrt N \hdz $. 

From \eqref{eq:dx}, we have 
 \begina
0 &=& \sumi Z_i \cdot \rp \cdot  \left\{x_i - \bar x - \hdo \right\},\\
0 &=&  \sumi (1-Z_i) \cdot \rp \cdot  \left\{x_i - \bar x- \hdz  \right\},
\enda
where $\hpi = p(x_i, Z_i; \hat\beta)$. 
This, together with Lemma \ref{lem:beta},  ensures $ (\delta_1, \delta_0, \mux, \beta) = (\hdo, \hdz, \bar x, \hat\beta)$ jointly solves 
\beginy\label{eq:psi_step1}
 0 = \sumi \phi(x_i, Z_i, \ryi; \theta)
\endy
with $\theta = (\delta_1, \delta_0, \mux, \beta)$ and 
\begina
 \phi =  
 \beginp 
 \phi_1\\
 \phi_0\\
 \psi_x\\
 \psi_\beta
 \endp
 \quad\text{with}\quad
 \begin{array}{lll}
\phi_1 &=& Z_i \cdot \rpfun \cdot \left(x_i - \mux - \delta_1\right),\\
\phi_0 &=& (1-Z_i) \cdot \rpfun \cdot \left(x_i - \mux - \delta_0\right),\\
\psi_x &=& x_i - \mux,\\
\psi_\beta &=& U_i \left\{\ryi - \pfun\right\}.
\end{array}
\enda
Recall that $\mux^*  = E(x_i)$ and $\beta^*$ is the true value of $\beta$ under Assumption \ref{assm:riy}. 
Direct algebra ensures that $\theta = \theta^* = (0, 0, \mux^*, \beta^*)$ solves $E\{\psi (\theta)\} = 0$. 
The theory of m-estimation ensures  
\begina
\sqrt N\left\{\beginp 
\hdo \\
\hdz  \\
\bar x\\
\hat\beta
\endp - \beginp 0\\0\\ \muxs \\ \beta^* \endp  \right\}  \rs  \mn\left\{ 0, \abax \right\},
\enda
where $A^*$ and $B^*$ are the values of $- E\left(\pd{(\delta_1, \delta_0, \mux, \beta)} \phi \right)$ and $E(\phi \phi^\T)$ evaluated at $\theta = \theta^*$. This completes step \eqref{step:2}.

\bigskip

\noindent\textbf{\underline{Step \eqref{step:3}.}}

By definition, we have 
\begina
\hywlsxop + \hdo^\T \gos =  \hyipwo =  \dfrac{\sumi Z_i \cdot \rp \cdot  Y_i(1)}{\sumi Z_i \cdot \rp} 
\enda
with 
\begina
0 = \sumi Z_i \cdot \rp \cdot \left\{ Y_i(1) - \hywlsxop  - \hdo^\T \gos\right\}.
\enda 
By symmetry, we also have 
\begina
0 = \sumi (1-Z_i) \cdot \rp \cdot \left\{ Y_i(0) - \hywlsxzp  - \hdz^\T \gzs\right\}.
\enda 
This, along with \eqref{eq:psi_step1}, ensures that 
$(\mu_1, \delta_1, \mu_0, \delta_0, \mux, \beta) = (\hywlsxop  , \hdo, \hywlsxzp   ,\hdz, \bar x, \hat\beta)$ jointly solves 
\begina
 0 = \sumi \psi\left( x_i, Y_i(1), Y_i(0), Z_i, \ryi; \theta\right),
\enda
with $\theta = (\mu_1, \delta_1, \mu_0, \delta_0, \mux, \beta)$ and 
\beginy\label{eq:ee_wlsx_true}
\psi =\beginp 
 \psi_1\\
 \phi_1\\
  \psi_0\\
 \phi_0\\
 \psi_x\\
 \psi_\beta
 \endp\quad\text{with}\quad 
 \begin{array}{lll}
\psi_1 &=& Z_i \cdot \rpfun \cdot \left\{ Y_i(1) - \mu_1  - \delta_1^\T \gos\right\},\\
\phi_1 &=& Z_i \cdot \rpfun \cdot \left(x_i - \mux - \delta_1\right),\\
\psi_0 &=& (1-Z_i) \cdot \rpfun \cdot \left\{ Y_i(0) - \mu_0  - \delta_0^\T \gzs\right\},\\
\phi_0 &=& (1-Z_i) \cdot \rpfun \cdot \left(x_i - \mux - \delta_0\right),\\
\psi_x &=& x_i - \mux,\\
\psi_\beta &=& U_i \left\{\ryi - \pfun\right\}.
\end{array}
\endy
Recall that $\muos  = E\{Y_i(1)\}$, $\muzs  = E\{Y_i(0)\}$, $\muxs = E(x_i)$, and $\beta^*$ is the true value of $\beta$ under Assumption \ref{assm:riy}. Direct algebra ensures that $\theta^* = (\muos, 0, \muzs, 0, \mux^*, \beta^*)$ solves $E\{\psi (\theta)\} = 0$. 
The theory of m-estimation ensures  
\beginy\label{eq:clt_wlsx_true}
\sqrt N\left\{\beginp 
\hywlsxop \\
\hdo \\
\hywlsxzp  \\
\hdz  \\
\bar x\\
\hat\beta
\endp - \beginp \muos\\ 0\\\muzs\\0\\ \muxs \\ \beta^* \endp  \right\}  \rs  \mn\left\{ 0, \ghg \right\},
\endy
where $G^* $ and $H^* $ are the values of $ - E\left(\pd{(\mu_1, \delta_1, \mu_0, \delta_0, \mux, \beta)} \psi \right)$ and $ E(\psi \psi^\T)$ evaluated at $\theta = \theta^*$. We compute below $G^*$ and $H^*$, respectively. 

\bigskip
\noindent\underline{\textbf{Compute  $H^*$.}}

Let $(\psi^*, \psi_1^*, \phi_1^*, \psi_0^*, \phi_0^*, \psi_x^*, \psi_\beta^*)$ denote the value of $(\psi, \psi_1 , \phi_1, \psi_0, \phi_0, \psi_x, \psi_\beta)$ evaluated at $\theta =\theta^*$. 
Recall $p_i = p(x_i, Z_i; \beta^*)$. 
From \eqref{eq:ee_wlsx_true}, we have 
\begina
\psi^*= \beginp 
 \psi^*_1\\
 \phi^*_1\\
  \psi^*_0\\
 \phi^*_0\\
 \psi^*_x\\
 \psi^*_\beta
 \endp\quad\text{with}\quad 
 \begin{array}{lll}
\psi_1^* &=& Z_i \cdot \rpn \cdot \dyo,\\
\phi_1^* &=& Z_i \cdot \rpn \cdot \left(x_i - \muxs \right),\\
\psi_0^* &=& (1-Z_i) \cdot \rpn \cdot \dyz,\\
\phi_0^* &=& (1-Z_i) \cdot \rpn \cdot \left(x_i - \muxs   \right),\\
\psi_x^* &=& x_i - \muxs,\\
\psi_\beta^* &=& U_i \left(\ryi - p_i\right),
\end{array}
\enda
where $\psi_1^*$, $\psi_0^*$, $\psi_x^*$, and $\psi_\beta^*$ coincide with those in \eqref{eq:psistar_wls} and \eqref{eq:psistar_xreg} in the proof of Theorem \ref{thm:clt_app}.
The $H^*$ matrix equals
\begina
H^* = \left.E\left(\psi \psi^\T\right)\right|_{\theta = \theta^*} = E\left(\psi^* \psi^{*\T}\right) =
\beginp
h_{11} & h_{12} & 0 & 0 & h_{15} & h_{16} \\
h_{12}^\T & h_{22} & 0 & 0 & h_{25} & h_{26} \\
0 & 0 & h_{33} & h_{34} & h_{35} & h_{36} \\
0 & 0 & h_{34}^\T & h_{44} & h_{45} & h_{46} \\
h_{15}^\T & h_{25}^\T & h_{35}^\T & h_{45}^\T & h_{55} & h_{56} \\
h_{16}^\T & h_{26}^\T & h_{36}^\T & h_{46}^\T & h_{56}^\T & h_{66} 
\endp 
\enda 
with 
\begina
h_{11} &=& E(\psi_1^* \psi_1^*)   = b_{11} \quad\text{by \eqref{eq:bs}} ,\\
h_{12} &=& E(\psi_1^* \phi_1^{*\T}) = E\left[Z_i \cdot \dfrac{\ryi}{p_i^2} \cdot \dyo\dxit \right] = E\left[Z_i \cdot p_i^{-1}  \cdot\dyo\dxit \right],\\
h_{15} &=& E(\psi_1^* \psi_x^{*\T}) \\
&=& E\left[Z_i \cdot \rpn  \cdot\dyo\dxit \right] = E\left[Z_i  \cdot \dyo\dxit \right] = e  \cdot \gamma_1^{*\T}\cov(x_i),\\
h_{16} &=& E(\psi_1^* \psi_\beta^{*\T})   = b_{13} \quad\text{by \eqref{eq:bs}},\\\\
h_{22} &=& E(\phi_1^* \phi_1^{*\T}) = E\left[Z_i \cdot \dfrac{\ryi}{p_i^2} \dxi\dxit \right] = E\left[Z_i \cdot p_i^{-1}  \cdot \dxi\dxit\right],\\
h_{25} &=& E(\phi_1^* \psi_x^{*\T}) = E\left[Z_i \cdot \rpn  \cdot\dxi\dxit \right] = E\left[Z_i  \cdot \dxi\dxit \right] = e  \cdot \cov(x_i),\\
h_{26} &=& E(\phi_1^* \psi_\beta^{*\T}) = E\left[Z_i \cdot \rpn  \cdot \dxi \cdot \left( \ryi - p_i \right) U_i^\T \right] \\
&=& E\left[Z_i  \cdot \rpn (1-p_i ) \cdot \dxi \cdot U_i^\T \right] = E\left[Z_i  \cdot  (1-p_i ) \cdot \dxi \cdot U_i^\T \right],
\enda
and
\begina
h_{55} &=&   E(\psi_x^* \psi_x^{*\T}) =  \cov(x_i),\\
h_{56} &=& E(\psi_x^* \psi_\beta^{*\T}) = E\left[ \dxi \cdot \left( \ryi - p_i \right) U_i^\T \right] = 0,\\
h_{66} &=& E(\psi_\beta^* \psi_\beta^{*\T}) = b_{33} \quad\text{by \eqref{eq:bs}}. 
\enda
We can compute $(h_{33}, h_{34}, h_{35}, h_{36})$ and $(h_{44}, h_{45}, h_{46})$ by symmetry.  
This ensures 
\beginy\label{eq:B_wlsx_true}
H^* =
\beginp
h_{11} & h_{12}  & 0 & 0 & e \cdot \gost \cov(x_i) & h_{16} \\
h_{12}^\T & h_{22} & 0 & 0 & e\cdot \cov(x_i)  & h_{26} \\
0 & 0 & h_{33} & h_{34} & (1-e) \cdot\gzst \cov(x_i) & h_{36} \\
0 & 0 & h_{34}^\T  & h_{44} & (1-e)\cdot \cov(x_i)  & h_{46} \\
e \cdot\gos \cov(x_i) &e  \cdot\cov(x_i)  & (1-e)  \cdot \gzs\cov(x_i) &  (1-e) \cdot \cov(x_i)  & \cov(x_i) & 0 \\
h_{16}^\T & h_{26}^\T & h_{36}^\T & h_{46}^\T & 0 & h_{66} 
\endp,  \quad 
\endy 
where 
\begina
h_{11} = b_{11}, \quad h_{16} = b_{13}, \quad h_{33} = b_{22}, \quad h_{36} = b_{23}, \quad h_{66} = b_{33}. 
\enda
\bigskip 

\noindent
\underline{\textbf{Compute   $G^*$.}}

With a slight abuse of notation, let  $\pd{\mu_1} \psi_1^*$ denote the value of $\pd{\mu_1} \psi_1 $ evaluated at $\theta=\theta^*$. Similarly define other partial derivatives. 
From \eqref{eq:ee_wlsx_true}, we have 
\renewcommand{\arraystretch}{1.5}
\begina
\left.\pd{(\mu_1, \delta_1, \mu_0, \delta_0, \mux, \beta)}\psi\right|_{\theta = \theta^*} = 
\beginp
\pd{\mu_1}\psi_1^* &  \pd{\delta_1^\T}\psi_1^*  & 0 & 0 & 0& \pd{\beta^\T} \psi_1^* \\
0  & \pd{\delta_1^\T}\phi_1^* & 0 & 0 & \pd{\mux^\T} \phi_1^* &\pd{\beta^\T} \phi_1^*\\
0 & 0 & \pd{\mu_0}\psi_0^* & \pd{\delta_0^\T}\psi_0^* & 0  & \pd{\beta^\T}\psi_0^*\\
0 & 0 & 0 & \pd{\delta_0^\T}\phi_0^* &  \pd{\mux^\T} \phi_0^* & \pd{\beta^\T}\phi_0^*\\
0 & 0 & 0 & 0 & \pd{\mux^\T}\psi_x^*  & 0\\
0 & 0 & 0 & 0 & 0 & \pd{\beta^\T}\psi_\beta^*
\endp
\enda 
with 
\begina
\pd{\delta_1^\T}\psi_1^* &=& - Z_i \cdot \rpn \gost,\\
 \pd{\delta_1^\T}\phi_1^* &=&  - Z_i \cdot \rpn \cdot I_J\\
 \pd{\mux^\T}\phi_1^* &=& - Z_i \cdot \rpn \cdot I_J,\\
  \pd{\beta^\T}\phi_1^* &=&- Z_i \cdot \rpn  (1-p_i) \cdot \left(x_i - \mux  \right)  \cdot U_i^\T. 
\enda
This ensures 
\begina
G^*& =& - \left. E \left( \pd{(\mu_1, \delta_1, \mu_0, \delta_0, \mux, \beta)}\psi \right) \right|_{\theta = \theta^*} =  
 -E \left(  \left.\pd{(\mu_1, \delta_1, \mu_0, \delta_0, \mux, \beta)}\psi\right|_{\theta = \theta^*} \right) \\
 &=&
 \left(
\begin{array}{cc|cc|cc}
g_{11} & g_{12} & 0 & 0 & 0& g_{16} \\
0& g_{22} & 0 &0 & g_{25} &g_{26}\\\hline
0 & 0 &g_{33} & g_{34} & 0 & g_{36}\\
0 & 0 & 0 & g_{44} & g_{45}  & g_{46}\\\hline
0 & 0 & 0 & 0 & g_{55}  & 0\\
0 & 0 & 0 & 0 & 0 & g_{66}
\end{array}\right)
\enda 
with 
\begina
g_{11} &=& -E\left( \pd{\mu_1}\psi_1^* \right) = e\quad\text{by \eqref{eq:as}},\\
g_{12} &=& -E\left( \pd{\delta_1^\T}\psi_1^*  \right) = E \left( Z_i \cdot \rpn\cdot \gost  \right) = e \cdot  \gost,\\
g_{16} &=& - E\left( \pd{\beta^\T} \psi_1^* \right) = b_{13} = h_{16} \quad\text{by \eqref{eq:as}},\\
g_{22} &=& - E\left( \pd{\delta_1^\T}\phi_1^* \right) = E \left( 
Z_i \cdot \rpn \cdot I_J \right)=  e\cdot I_J,\\
g_{25} &=& - E\left( \pd{\mux^\T}\phi_1^* \right) = E \left( 
Z_i \cdot \rpn \cdot I_J \right)= e\cdot I_J,\\
g_{26} &=& - E\left( \pd{\beta^\T}\phi_1^* \right) = E \left\{ Z_i \cdot  (1-p_i) \cdot \dxi  \cdot U_i^\T\right\} = h_{26},\\
g_{55} &=&- E \left(\pd{\mux^\T }\psi_x^* \right) = I_J,\\
g_{66} &=& - E\left(\pd{\beta}\psi_\beta^*\right) = b_{33} = h_{66}. 
\enda
We can similarly compute $(g_{33}, g_{34}, g_{36})$ and $(g_{44}, g_{45}, g_{46})$ by symmetry. 
This ensures 
\begina
G^* = \left( \begin{array}{cccc|cc}
e & e\cdot \gost  & 0  & 0 & 0& h_{16} \\
0 & e \cdot I_J  &0 & 0 & e \cdot I_J & h_{26}\\
0&  0 &(1-e) & (1-e) \cdot \gzst & 0 &h_{36}\\
0 & 0 & 0 & (1-e)\cdot I_J   & (1-e)\cdot I_J   & h_{46}\\\hline
0 & 0 & 0 & 0 & I_J & 0\\
0 & 0 & 0 & 0 & 0 & h_{66}
\end{array}\right) 
= \left(\begin{array}{c| c}
G_{(11)} & G_{(12)} \\\hline 
0 & G_{(22)}
\end{array}\right) 
\enda
with 
\beginy\label{eq:G_wlsx_true}
G^{*-1} 
&=&  \beginp
G_{(11)}^{-1} &  -G_{(11)}^{-1}G_{(12)} G_{(22)}^{-1} \\
0 & G_{(22)}^{-1}
\endp\nonumber\\
&=& \left(\begin{array}{cc cc|c c}
e^{-1} & - e^{-1} \gost & 0 & 0 & \gost  & -e^{-1}\left( h_{16} - \gost h_{26}\right) b^{-1}_{66} \\
0 &e^{-1} I_J  & 0 & 0 & -I_J &- e^{-1} h_{26} b^{-1}_{66} \\ 
0 &0 &(1-e)^{-1} & -(1-e)^{-1}\gzst & \gzst & -(1-e)^{-1}\left( h_{36} - \gzst h_{46}\right) b^{-1}_{66}\\
0 &0 & 0&(1-e)^{-1}I_J & - I_J  & -(1-e)^{-1} h_{46}b^{-1}_{66} \\\hline
0 &0& 0 & 0 & I_J &0\\ 
0& 0 & 0 & 0 & 0 & b^{-1}_{66} \\
\end{array}\right).
\endy
Equations \eqref{eq:clt_wlsx_true}--\eqref{eq:G_wlsx_true} together complete step \eqref{step:3}.

\bigskip 
\bigskip 
\noindent\underline{\textbf{Step \eqref{step:4}.}}
%
%
Recall 
\begina
\cova(\htreg , \Delta)  
&=& \cova\left\{\hywlsxop , \ \hdo  \right\} \gos -\cova\left\{\hywlsxop   , \  \hdz \right\}\gzs \\
&&- \cova\left\{ \hywlsxzp , \ \hdo   \right\} \gos+\cova\left\{ \hywlsxzp , \ \hdz   \right\}\gzs 
\enda
from \eqref{eq:true_cov}. 
We compute in the following the four asymptotic covariances on the right-hand side one by one. 

First, it follows from 
 \begina
 \hywlsxop =(1, 0, 0, 0, 0, 0 ) 
\beginp 
\hywlsxop \\
\hdo \\
\hywlsxzp  \\
\hdz  \\
\bar x\\
\hat\beta
\endp, \quad 
\hdo  = (0 , I_{J}, 0 , 0 , 0 , 0  )\beginp 
\hywlsxop \\
\hdo \\
\hywlsxzp  \\
\hdz  \\
\bar x\\
\hat\beta
\endp 
\enda
that 
\begina
\cova\left\{\hywlsxop , \ \hdo  \right\}  =  (1, 0, 0, 0, 0, 0) \cdot \ghg \cdot \beginp
0\\
I \\
0\\
0\\
0\\
0
\endp,
\enda
where 
\begina
 (1, 0, 0, 0, 0, 0 ) \cdot G^{*{-1}} &=& 
\beginp
e^{-1}, & - e^{-1} \gost, & 0, & 0, & \gost,  & -e^{-1}\left( h_{16} - \gost  h_{26}\right) b^{-1}_{66} 
\endp,\nonumber\\
 (0 , I_{J}, 0 , 0 , 0 , 0  ) \cdot G^{*{-1}} &=& \beginp
0, &e^{-1} I_J,  & 0, & 0, & -I_J, & -e^{-1} h_{26} b^{-1}_{66}
\endp 
\enda
from  \eqref{eq:G_wlsx_true}. 
This, together with \eqref{eq:B_wlsx_true}, ensures
\beginy\label{eq:true_cov11}
&& \cova\left\{\hywlsxop , \ \hdo  \right\} \nonumber\\
&=& \beginp
e^{-1}, & - e^{-1} \gost, & 0, & 0, & \gost,  & -e^{-1}\left( h_{16} - \gost h_{26}\right) b^{-1}_{66} 
\endp\nonumber\\
&&\cdot \beginp
h_{11} & h_{12}  & 0 & 0 & e \cdot \gost\cov(x_i) & h_{16} \\
h_{12}^\T & h_{22} & 0 & 0 & e\cdot \cov(x_i)  & h_{26} \\
0 & 0 & h_{33} & h_{34} & (1-e) \cdot \gzst\cov(x_i) & h_{36} \\
0 & 0 & h_{34}^\T  & h_{44} & (1-e)\cdot \cov(x_i)  & h_{46} \\
e \cdot \gos \cov(x_i)&e\cdot \cov(x_i)  & (1-e) \cdot \gzs\cov(x_i) &  (1-e)\cdot \cov(x_i)  & \cov(x_i) & 0 \\
h_{16}^\T & h_{26}^\T & h_{36}^\T & h_{46}^\T & 0 & h_{66} 
\endp 
\beginp
0\\
e^{-1} I_J \\
0 \\
0\\
-I_J\\
-e^{-1} h_{66}^{-1}h_{26}^\T
\endp \nonumber\\
&=& 
 \beginp
e^{-1}, & - e^{-1} \gost, & 0, & 0, & \gost,  & -e^{-1}\left( h_{16} - \gost h_{26}\right) b^{-1}_{66} 
\endp
\beginp
e^{-1}h_{12} - e\cdot \gost  \cov(x_i) - e^{-1}h_{16} h_{66}^{-1}h_{26}^\T\\
e^{-1}h_{22} - e\cdot\cov(x_i) - e^{-1} h_{26}h_{66}^{-1}h_{26}^\T\\
-(1-e)\cdot \gzst \cov(x_i) - e^{-1} h_{36}h_{66}^{-1}h_{26}^\T\\
-(1-e)\cdot \cov(x_i) - e^{-1} h_{46}h_{66}^{-1}h_{26}^\T\\
0\\
0
\endp \nonumber\\
&
=& e^{-2}h_{12} - \gost\cov(x_i) - e^{-2} h_{16}h_{66}^{-1} h_{26}^\T - e^{-2}\gamma_1^\T h_{22} + \gost\cov(x_i) + e^{-2} \gost h_{26}h_{66}^{-1} h_{26}^\T\nonumber\\
&
=& e^{-2}\left\{ \left( h_{12}    -  \gost h_{22}  \right) - \left( h_{16}   -   \gost h_{26}\right)h_{66}^{-1} h_{26}^\T\right\} .
\endy

Similarly, 
\beginy\label{eq:true_cov01}
&& \cova\left\{\hywlsxzp , \ \hdo  \right\}\nonumber \\
&=& \beginp
0, &0, &(1-e)^{-1}, & -(1-e)^{-1}\gzst, & \gzst, & -(1-e)^{-1}\left( h_{36} - \gzst h_{46}\right) b^{-1}_{66} \endp \nonumber \\
&&\cdot \beginp
h_{11} & h_{12}  & 0 & 0 & e \cdot \gost\cov(x_i) & h_{16} \\
h_{12}^\T & h_{22} & 0 & 0 & e\cdot \cov(x_i)  & h_{26} \\
0 & 0 & h_{33} & h_{34} & (1-e) \cdot \gzst\cov(x_i) & h_{36} \\
0 & 0 & h_{34}^\T  & h_{44} & (1-e)\cdot \cov(x_i)  & h_{46} \\
e \cdot \gos \cov(x_i)&e\cdot \cov(x_i)  & (1-e) \cdot \gzs\cov(x_i) &  (1-e)\cdot \cov(x_i)  & \cov(x_i) & 0 \\
h_{16}^\T & h_{26}^\T & h_{36}^\T & h_{46}^\T & 0 & h_{66} 
\endp 
\beginp
0\\
e^{-1} I_J \\
0 \\
0\\
-I_J\\
-e^{-1} h_{66}^{-1}h_{26}^\T
\endp\nonumber \\
&=& 
\beginp
0, &0, &(1-e)^{-1}, & -(1-e)^{-1}\gzst, & \gzst, & -(1-e)^{-1}\left( h_{36} - \gzst h_{46}\right) b^{-1}_{66} \endp \nonumber \\
&&
\cdot 
\beginp
e^{-1}h_{12} - e\cdot \gost  \cov(x_i) - e^{-1}h_{16} h_{66}^{-1}h_{26}^\T\\
e^{-1}h_{22} - e\cdot\cov(x_i) - e^{-1} h_{26}h_{66}^{-1}h_{26}^\T\\
-(1-e)\cdot \gzst \cov(x_i) - e^{-1} h_{36}h_{66}^{-1}h_{26}^\T\\
-(1-e)\cdot \cov(x_i) - e^{-1} h_{46}h_{66}^{-1}h_{26}^\T\\
0\\
0
\endp \nonumber\\
&
=&   - \gzst \cov(x_i) - (1-e)^{-1}e^{-1} h_{36}h_{66}^{-1} h_{26}^\T  + \gzst\cov(x_i) + (1-e)^{-1}e^{-1} \gzst h_{46}h_{66}^{-1} h_{26}^\T\nonumber\\
&
=&    - (1-e)^{-1}e^{-1} \left(  h_{36}  - \gzst h_{46}\right)h_{66}^{-1} h_{26}^\T.
\endy

Equations \eqref{eq:true_cov11}--\eqref{eq:true_cov01} ensure
\begina
 \cova\left\{\hywlsxop , \ \hdo  \right\} 
&
=& e^{-2} \left\{ \left( h_{12}    -  \gost h_{22}  \right) - \left( h_{16}   -   \gost h_{26}\right)h_{66}^{-1} h_{26}^\T\right\},\\
\cova\left\{\hywlsxzp , \ \hdo  \right\} 
&
=& - (1-e)^{-1}e^{-1} \left(  h_{36}  - \gzst  h_{46}\right)h_{66}^{-1} h_{26}^\T.
\enda
By symmetry, we have
\begina
\cova\left\{\hywlsxzp , \ \hdz  \right\}  
&
=& (1-e)^{-2} \left\{\left( h_{34}   -\gzst h_{44} \right)-\left(  h_{36}- \gzst h_{46}\right)h_{66}^{-1} h_{46}^\T\right\},\\
\cova\left\{\hywlsxop , \ \hdz  \right\} 
&
=&- (1-e)^{-1}e^{-1} \left(  h_{16}   - \gost h_{26}\right)h_{66}^{-1} h_{46}^\T.
\enda
 When the true outcome models are  indeed linear as $E\{Y(z) \mid x_i\} = \mu_z + (x_i-\mu_x)^\T\gamma_z \ (z=0,1)$,  we have 
 \begina
 h_{12} = \gost h_{22},\qquad h_{16} =  \gost h_{26}, \qquad 
 h_{34} = \gzst h_{44},\qquad h_{36} = \gzst h_{46}
 \enda
such that 
\begina
 \cova\left\{\hywlsx'(z) , \ \hywlsx'(z')  \right\} =  0\quad \text{for all $z, z' = 0,1$}. 
\enda
This ensures $
\cova(\htreg , \Delta) = 0$ 
by \eqref{eq:true_cov}.

\section{Proof of the rest of the results}\label{sec:proof_other}
\subsection{Proof of Theorem \ref{thm:ipwx_clt_xy} and Proposition \ref{prop:invariance} in the main paper}\label{sec:main}

\begin{proof}[Proof of Theorem \ref{thm:ipwx_clt_xy}]
Given $\riw \indep Z_i$, $x_i^\aug = (x_i , \riw, \wiz)$ and $x_i^\sub$, as a subset of $x_i^\aug$, are both independent and identically distributed across all $N$ units and unaffected by the treatment assignment. 
We can hence view them as fully observed baseline covariates, analogous to the $x_i$ in Assumption \ref{assm:riy}. 
The result follows from Proposition \ref{prop:eps_eff} after renewing $x_i^\mim$ as our new fully observed covariates.
\end{proof}

\begin{proof}[Proof of Proposition \ref{prop:invariance}]
Observe that the column space of $\{x_i^\mim(c)\}_{i=1}^N$ is identical to the column space of $(x_i^\mim)_{i=1}^N$. 
From Lemma \ref{lem:mle}, the $\hat p_i$ estimated by $
\glmt\{\riy \sim 1 + x_i^\mim(c) + Z_i + x_i^\aug(c) Z_i\} $ and the $\hei $ estimated by $
\glmt\{Z_i \sim 1 + x_i^\mim(c)\}$ are constant across all $c \in \mathbb R^K$. This ensures the invariance of $\hteps(x_i^\mim(c))$ across all choices of $c$.  
\end{proof}
\subsection{Proofs of the additional results in Section \ref{sec:additional results}}\label{sec:additional}

The proof of Proposition \ref{prop:hajek_app} follows from standard properties of least squares. We omit the details.

\begin{proof}[Proof of Proposition \ref{prop:wlsx_dr}]
Observe that 
$\hat m_1(x_i)$ and $\hat m_0(x_i)$ coincide with the estimates of the outcome models from the weighted-least-squares fits of \eqref{eq:wls_x_1_2} and \eqref{eq:wls_x_0_2}, respectively, 
with 
\beginy\label{eq:mhat}
\hat m_1(x_i) = \hat Y_\xreg(1) + (x_i-\bar x)^\T\hat\gamma_1, \quad \hat m_0(x_i) =  \hat Y_\xreg(0) + (x_i-\bar x)^\T\hat\gamma_0.
\endy
The first-order conditions of \eqref{eq:wls_x_1_2} and \eqref{eq:wls_x_0_2} ensure 
\begina
 \sumi Z_i \cdot \rp \cdot \left\{Y_i(1)-\hmo(x_i) \right\}  =  
\sumi (1-Z_i) \cdot \rp \cdot \left\{Y_i(0) - \hmz(x_i) \right\}  = 0.
\enda
This ensures 
\begina
\hy_\dr(1)= \meani \hat m_1(x_i)\overset{\eqref{eq:mhat}}{=} \hat Y_\xreg(1), \quad \hy_\dr(0)= \meani \hat m_0(x_i)\overset{\eqref{eq:mhat}}{=} \hat Y_\xreg(0)
\enda
with 
\begina
\hy_\dr(1) - \hy_\dr(0) 
= \hat Y_\xreg(1) - \hat Y_\xreg(0)\overset{\text{Lemma \ref{lem:xreg_num}}}{=} \htreg.
\enda

\end{proof}

The proof of Proposition \ref{prop:reg-ps} is almost identical to the proof of Proposition \ref{prop:wlsx_dr} and omitted.

\begin{proof}[Proof of Proposition \ref{prop:eps_eff}]  
Almost identical proof as that for $\hteps$ in Section \ref{sec:ipwx_proof} ensures that 
$\htau_{\ipwx}'$ is consistent and  asymptotically normal with asymptotic variance 
$
\va(\hteps') = v_\ua - e(1-e)\var \{ \textup{proj}(\tyi \mid 1, x'_i)  \}
$. 
The asymptotic efficiency of $\hteps$ over $\hteps'$ then follows from $\var \{ \textup{proj}(\tyi \mid 1, x'_i)  \} \leq \var \{ \textup{proj}(\tyi \mid 1, x_i)  \}$ by properties of linear projection. 

In particular, given $x_i'$ as the alternative covariate vector that may be of different dimension from $x_i$,  let $\hat\alpha'$ denote the coefficient vector from $\glmt(Z_i \sim 1 +x'_i)$ over $i = \ot{N}$, including the intercept; let $\alpha'$ denote the new parameter of \eqref{eq:e_logit} corresponding to $x_i'$; let $\alpha^{*\prime}$ denote the true value of the parameter under Bernoulli randomization with $e(x_i', \alpha^{*\prime}) = e$. 
By similar reasoning as the proof of $\hteps$  in Section \ref{sec:ipwx_proof}, we have  
$\hteps' = \hy_\xeps(1)' - \hy_\xeps(0)'$ and  $(\mu_1, \mu_0, \beta, \alpha') =(\hyipwxo' , \hyipwxz'  , \hat\beta, \hat\alpha')$ jointly solves 
\begina
0 = \meani \phi' \big(Y_i(1), Y_i(0), x_i, Z_i, \ryi; \mu_1, \mu_0, \beta, \alpha \big),
\enda
where \renewcommand{\arraystretch}{1.5}
\beginy\label{eq:ee_ipwx}
\phi' = \beginp  \phi'_1 \\ \phi'_0 \\ \psi_\beta\\\psi'_\alpha\endp 
\quad\text{with} \quad
\begin{array}{ccl}
\phi_1' &=& 
 \rpfun  \cdot \dfrac{Z_i}{e(x_i', \alpha')} \cdot \big\{Y_i(1) - \mu_1 \big\}, \\
\phi_0' &=& \rpfun  \cdot \dfrac{1-Z_i}{1- e(x_i', \alpha')} \cdot \big\{Y_i(0) - \mu_0 \big\},\\
\psi_\beta &=&  U_i \big\{\riy - \pfun \big\},\\
\psi_\alpha'&=& \txi' \big\{Z_i - e(x_i', \alpha') \big\}.
\end{array}
\endy
The asymptotic distribution of $\hteps'$ then follows from almost identical m-estimation calculations as that for $\hteps$. 
\end{proof}

\end{document}